\documentclass{article}
\usepackage{euscript,amsmath,amssymb,amsfonts,graphicx,bm,amsthm,mathtools}
\DeclarePairedDelimiter\ceil{\lceil}{\rceil}
\DeclarePairedDelimiter\floor{\lfloor}{\rfloor}

\usepackage{cite} 

  \usepackage{paralist}
  \usepackage{graphics} 
\usepackage[caption=false]{subfig}
\usepackage{soul}

\usepackage{latexsym,amssymb,enumerate,amsmath,amsthm} 
  \usepackage{graphicx} 
  \usepackage{tikz}
     \usepackage[colorlinks=false]{hyperref}
 \hypersetup{urlcolor=blue, citecolor=red}
\usepackage{latexsym,amssymb,enumerate,amsmath} 
\usepackage{pgfplots}
\usepackage{soul,xcolor}

\newtheorem{theorem}{Theorem}

\newtheorem{corollary}[theorem]{Corollary}
\theoremstyle{plain}
\theoremstyle{remark}

\theoremstyle{definition}
\newtheorem{definition*}{Definition}

\def\blue{}
\newcommand{\dd}{\normalfont\textup{d}}
\def\eps{\varepsilon}
\def\E{\mathbb{E}}
\def\P{\mathbb{P}}
\def\R{\mathbb{R}}
\def\Z{\mathbb{Z}}
\def\A{A}
\def\single{\normalfont\textup{single}}
\def\double{\normalfont\textup{double}}
\def\perf{\normalfont\textup{perf}}
\def\ka{\normalfont k_{\textup{a}}}
\def\ke{\normalfont k_{\textup{e}}}
\def\AUC{\normalfont\textup{AUC}}
\def\XR{\normalfont\textup{XR}}
\def\IR{\normalfont\textup{IR}}
\def\C{\normalfont C_{0}}

\begin{document}


\title{Designing drug regimens that mitigate nonadherence}


\author{Elijah D. Counterman\thanks{Department of Mathematics, University of Utah, Salt Lake City, UT 84112 USA.} \and Sean D. Lawley\thanks{Department of Mathematics, University of Utah, Salt Lake City, UT 84112 USA (\texttt{lawley@math.utah.edu}).}
}
\date{\today}
\maketitle

\begin{abstract}
Medication adherence is a well-known problem for pharmaceutical treatment of chronic diseases. Understanding how nonadherence affects treatment efficacy is made difficult by the ethics of clinical trials that force patients to skip {\blue doses of the medication being tested}, the unpredictable timing of missed doses by actual patients, and the many competing variables that can either mitigate or magnify the deleterious effects of nonadherence, such as pharmacokinetic absorption and elimination rates, dosing intervals, dose sizes, adherence rates, etc. In this paper, we formulate and analyze a mathematical model of the drug concentration in an imperfectly adherent patient. Our model takes the form of the standard single compartment pharmacokinetic model with first order absorption and elimination, except that the patient takes medication only at a given proportion of the prescribed dosing times. Doses are missed randomly, and we use stochastic analysis to study the resulting random drug level in the body. We then use our mathematical results to propose principles for designing drug regimens that are robust to nonadherence. In particular, we quantify the resilienc{\blue e} of extended release drugs to nonadherence, which is quite significant in some circumstances, and we show the benefit of taking a double dose following a missed dose if the drug absorption or elimination rate is slow compared to the dosing interval. We further use our results to compare some antiepileptic and antipsychotic drug regimens. 
\end{abstract}


\section{Introduction}

Adherence to a medication regimen refers to the extent to which patients take medications as prescribed by their physicians \cite{osterberg2005}. Nonadherence to medication is an age-old problem, as even Hippocrates warned physicians to ``keep watch also on the fault of patients which makes them lie about the taking of things prescribed'' \cite{brown2013}. Today in the United States, it is estimated that medication nonadherence accounts for up to 25\% of hospitalizations, 50\% of treatment failures, and around 125,000 deaths per year \cite{kim2018}. Remarkably, the World Health Organization has claimed that improving adherence may have a far greater impact on public health than any improvement in specific medical treatments \cite{who2003, haynes2002}.

Strategies for improving adherence generally focus on (i) communication between the patient and the healthcare provider, (ii) simplifying drug regimens, and (iii) choosing medications which maintain therapeutic effect despite lapses in adherence \cite{osterberg2005}. While adherence is difficult to quantify, adherence is usually reported as the proportion
\begin{align}\label{p0}
p\in[0,1]
\end{align}
of doses of medication actually taken by the patient over time \cite{osterberg2005}. Hence, strategies (i) and (ii) aim to increase the adherence $p$, whereas strategy (iii) seeks to improve treatment efficacy given a $p<1$.

Drugs which are absorbed by the body slowly, or eliminated from the body slowly, are sometimes suggested to accomplish strategy (ii), since they can allow less frequent dosing. An important class of such drugs are called extended release (XR) drugs \cite{gidal2021, wheless2018, vadivelu2011}, which are defined by a very slow release rate compared to their immediate release (IR) counterpart. XR drugs are sometimes called sustained release, slow release, or controlled release \cite{qiu2011}.

To explain more precisely, we recall the standard pharmacokinetic model of oral administration in a single compartment with first order absorption and elimination \cite{gibaldi1982, bauer2015}. In this model, a patient takes a dose of size $D>0$ every $\tau>0$ units of time, and doses are absorbed into the body at rate $\ka>0$ and eliminated at rate $\ke>0$. XR and IR formulations of a drug have identical values of the elimination rate, but the XR absorption rate is much slower than the IR absorption rate. In addition, the XR dosing interval is sometimes larger than the IR dosing interval, with a commensurate larger dose size for the larger dosing interval. For example, XR versions are often dosed once-daily rather than twice-daily \cite{gidal2021}. Using superscripts to denote the XR and IR formulations, this means that
\begin{align}\label{s0}
\ke^{\IR}
=\ke^{\XR},\quad
\ka^{\IR}
\gg\ka^{\XR},\quad
\tau^{\IR}
<\tau^{\XR},\quad
D^{\IR}
<D^{\XR}.
\end{align}

However, the efficacy of XR versus IR drugs when confronted with nonadherence is controversial. While some have emphasized benefits of less frequent dosing allowed by XR drugs, others have sought to alert clinicians of the increased danger following missed doses of XR drugs \cite{pellock2004, bialer2007, levy1993, perucca2009}. Indeed, in a critical review of the use of XR drugs for the treatment of epilepsy, Bialer \cite{bialer2007} warned of the increased risk of breakthrough seizure following a missed dose of a once-daily drug compared to a twice-daily drug.

These two competing views can be understood in terms of the parameters in \eqref{p0} and \eqref{s0}. On one side, decreasing the absorption rate $\ka$ allows for more stable drug concentrations in the body between doses, and increasing the dosing interval $\tau$ is associated with higher adherence rates $p$. On the other side, increasing $\tau$ causes larger fluctuations in the drug concentration in the body between doses, and increasing $\tau$ and the dose size $D$ exacerbates the effects of a missed dose. To weigh these various factors and evaluate these competing views, a more precise, quantitative investigation is needed. 

Another disputed question related to nonadherence is what patients should do following a missed dose. Should patients skip the missed dose or should they take a double dose to compensate? Despite the fact that this is one of the most common questions that patients ask, they often do not receive adequate instructions for what to do after missing a dose \cite{howard1999, albassam2020, gilbert2002}. Most answers to this question focus on the role of the (terminal) drug half-life $t_{1/2}$, which is related to the elimination rate {\blue in the single compartment model described above} via
\begin{align*}
t_{1/2}
=\frac{\ln2}{\ke}
\approx\frac{0.69}{\ke}.
\end{align*}
Curiously, long half-lives are sometimes cited as the reason to skip a missed dose \cite{albassam2020, gilbert2002} and sometimes cited as the reason to take a double dose to make up for a missed dose \cite{jonklaas2014}. To our knowledge, the role of the absorption rate $\ka$ in the proper handling of missed doses has not been investigated. 

In this paper, we formulate and analyze a mathematical model for the drug concentration in an imperfectly adherent patient. Our model takes the form of the standard single compartment pharmacokinetic model with first order absorption and elimination described above, except that the patient takes medication only at a given proportion $p$ of the prescribed dosing times. Doses are missed randomly, and we study the resulting random drug level in the body. We find explicit formulas for various pharmacologically important statistics, including (a) the average drug level {\blue in} the patient and (b) how the drug levels in the patient deviate from the average level in a perfectly adherent patient, which we call the error. We then use these formulas to investigate how drug levels depend on the various parameters and on how missed doses are handled. In particular, we investigate the effects of skipping a missed dose versus taking a double dose to compensate, which we refer to respectively as the single dose protocol and the double dose protocol. The model generalizes our previous model in \cite{counterman2021} which assumed an infinite absorption rate (i.e.\ $\ka=\infty$). The model is illustrated in Figure~\ref{figschem}.

\begin{figure}[t]
\centering
\includegraphics[width=1\linewidth]{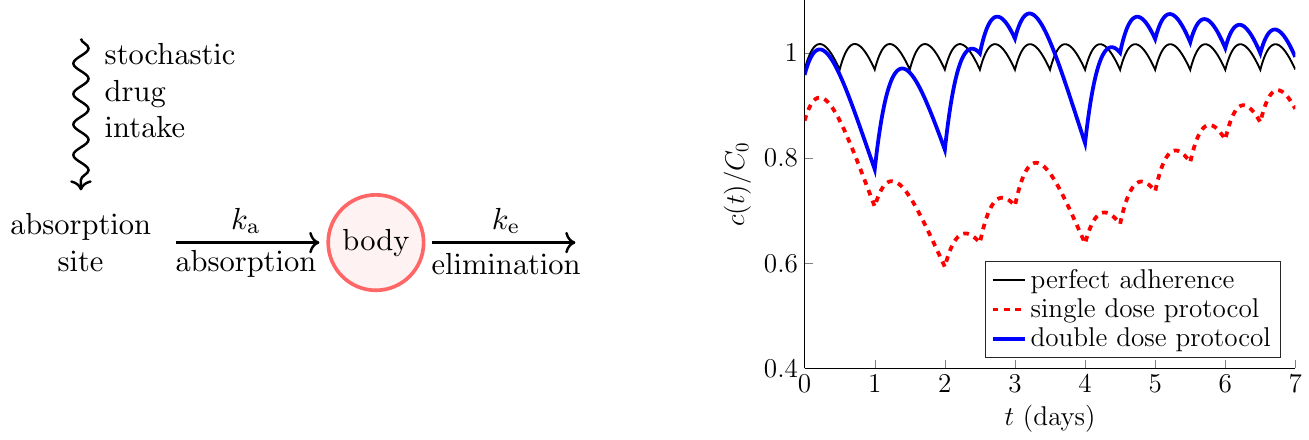}
\caption{The left panel shows a schematic diagram of the model, which involves the standard pharmacokinetic model in which the body is modeled by a single compartment and the drug is absorbed into the body from the absorption site at rate $\ka$ and eliminated at rate $\ke$. We model nonadherence by a stochastic drug intake in which the patient randomly misses doses. The right panel plots sample time courses of the drug concentration for perfect adherence (thin black curve) and imperfect adherence for the single dose protocol (dashed red curve) and the double dose protocol (thick blue curve).}
\label{figschem}
\end{figure}

Our analysis has several salient pharmacological implications for designing drug regimens to mitigate nonadherence. 
First, the double dose protocol substantially increases the average drug level compared to the single dose protocol for typical adherence rates $p$. Second, slower absorption rates, slower elimination rates, and shorter dosing intervals all decrease the drug level {error}s. Third, as long as the absorption and elimination rates are not both fast compared to the dosing interval, we find that the double dose protocol has markedly smaller {error}s than the single dose protocol. Furthermore, in this case of slow absorption and/or slow elimination, we show that the double dose protocol cannot cause the drug levels in the body to rise much above the average drug level in a perfectly adherent patient. Finally, we show how our results can be used to compare specific drug regimens and design XR drugs to deal with the challenge of nonadherence.

The rest of the paper is organized as follows. We formulate the model in section~\ref{model} and analyze it in section~\ref{math}. Mathematically, the model involves generalizing a class of random variables whose rather exotic distributions have been studied in the pure mathematics literature for many years, dating back to Paul Erd\H{o}s and others in the 1930s \cite{jessen1935, kershner1935, erdos1939}. In section~\ref{pharm}, we explore some pharmacological implications of the model. While the general mathematical analysis in section~\ref{math} is flexible to allow for correlations in missed doses (i.e.\ the patient can be more or less likely to miss a dose following a missed dose(s)), for simplicity we focus in section~\ref{pharm} on the case that the patient misses doses independently of their prior behavior. We note, however, that the actual number of doses taken by the patient at different dosing times is not independent for the double dose protocol analyzed in section~\ref{pharm}. In the Discussion section, we discuss our results in the context of related work and some specific XR and IR drugs. {\blue We also discuss how our results could be tested in clinical trials.} The proofs of the theorems are presented in the Appendix.

\section{Mathematical model}\label{model}

We begin by reviewing the standard pharmacokinetic model of oral administration in a single compartment with first order absorption and elimination \cite{gibaldi1982, bauer2015}. The drug concentration in the body at time ${t}>0$ is denoted by $c(t)$ and it evolves according to the following ordinary differential equation (ODE),
\begin{align}\label{code}
\frac{\dd c}{\dd {t}}
=\ka\frac{g}{V}-\ke c,
\end{align}
where $\ka$ is the absorption rate, $\ke$ is the elimination rate, $V$ is the volume of distribution, and $g$ is the amount of the drug at the absorption site. The amount $g$ follows the ODE,
\begin{align}\label{gode}
\frac{\dd g}{\dd {t}}
=-\ka g+I({t}),
\end{align}
where $I({t})$ is the drug input, which depends on the timing and sizes of drug doses actually taken by the patient. In the standard model of perfect adherence, $I({t})$ is a deterministic sum of Dirac delta functions. In our model of imperfect adherence, $I({t})$ is a random sum of Dirac delta functions, and we study the probability distribution of the resulting random drug concentration $c({t})$ that is subject to this random forcing $I({t})$. In our previous work \cite{counterman2021}, we studied this model in the simplified case in which $\ka=\infty$.

\subsection{Perfect adherence}

Suppose a dose of size $D>0$ is prescribed at regular time intervals of length $\tau>0$. The drug input for perfect adherence is
\begin{align}\label{iperf}
I^{\perf}({t})
=DF\sum_{n\ge0}\delta_{\textup{dirac}}({t}-n\tau),
\end{align}
where $F\in(0,1]$ is the bioavailability and $\delta_{\textup{dirac}}$ is the Dirac delta function. Solving the ODEs \eqref{code}-\eqref{iperf} yields the drug concentration at time ${t}\ge0$ in the perfectly adherent patient \cite{bauer2015},
\begin{align}\label{cstart}
c^{\perf}({t})
:=
\kappa\sum_{n=0}^{N({t})}\Big(e^{-\ke({t}-n\tau)}-e^{-\ka({t}-n\tau)}\Big),
\end{align}
where $\kappa$ is the (possibly negative) concentration scale,
\begin{align}\label{concscale}
\kappa
:=\frac{DF}{V}\frac{\ka}{\ka-\ke}\in\R,
\end{align}
and $N({t})+1$ is the number of dosing times which occur before time ${t}$,
\begin{align*}
N({t})
:=\sup\{n\ge0:n\le {t}/\tau\}.
\end{align*}
We assume $\ka\neq\ke$ throughout this paper.

\subsection{Nonadherence}

We now suppose that at each dosing time, the patient ``remembers'' or ``forgets'' to take medication with respective probabilities $p\in(0,1]$ and $1-p$. To describe this more precisely, let $\{\xi_{n}\}_{n}$ be a sequence of Bernoulli random variables with parameter $p$, which means
\begin{align}\label{xin}
\begin{split}
\xi_{n}
=\begin{cases}
1 & \text{with probability }p,\\
0 & \text{with probability }1-p.
\end{cases}
\end{split}
\end{align}
The event $\xi_{n}=1$ means that the patient takes medication at the $n$th dosing time. Our mathematical analysis in section~\ref{math} allows $\xi_{n}$ and $\xi_{m}$ to be dependent (see section~\ref{math} for precise assumptions), which means that the patient can be more or less likely to remember (or forget) at the $m$th dosing time if they remember (or forget) at the $n$th dosing time. However, for simplicity the applications in section~\ref{pharm} focus on the case that $\{\xi_{n}\}_{n}$ are independent.

Letting $Df_{n}\ge0$ denote the drug amount taken at the $n$th dosing time, the drug input is
\begin{align}\label{iimp}
I({t})
=DF\sum_{n\ge0}\delta_{\textup{dirac}}({t}-n\tau)f_{n}.
\end{align}
Solving the ODEs \eqref{code}-\eqref{gode} with the drug input given by \eqref{iimp} yields the drug concentration in the imperfectly adherent patient,
\begin{align}\label{CC0}
c({t})
=
\kappa\sum_{n=0}^{N({t})}\Big(e^{-\ke({t}-n\tau)}-e^{-\ka({t}-n\tau)}\Big)f_{n}.
\end{align}
Note that \eqref{CC0} becomes \eqref{cstart} if $f_{n}=1$ for all $n$. 
Since the patient cannot take medication when they forget, it is natural to take
\begin{align*}
f_{n}=0,\quad\text{if }\xi_{n}=0.
\end{align*}
However, since the patient may take more than a single dose to make up for prior missed doses, we allow
\begin{align*}
f_{n}> 1,\quad \text{if }\xi_{n}=1.
\end{align*}
In the general analysis below, $f_{n}$ is a function of the history $\{\xi_{i}\}_{i=n-m}^{n}$ for some $m\ge0$, and we refer to a choice of $f_{n}$ as a ``dosing protocol.''

A common dosing protocol is for the patient to simply take a single dose when they remember, which means
\begin{align}
\label{fsingle}
f_{n}
=f_{n}^{\textup{single}}
:=
\begin{cases}
0 & \text{if }\xi_{n}=0,\\
1 &  \text{if }\xi_{n}=1.
\end{cases}
\end{align}
We call \eqref{fsingle} the ``single dose protocol.'' Another simple dosing protocol is for the patient to take a double dose to make up for a missed dose at the prior dosing time, which means
\begin{align}
\begin{split}\label{fdouble}
f_{n}
=f_{n}^{\textup{double}}
:=\begin{cases}
0 & \text{if }\xi_{n}=0,\\
1 & \text{if }\xi_{n}=1,\,\xi_{n-1}=1,\\
2 & \text{if }\xi_{n}=1,\,\xi_{n-1}=0.
\end{cases}
\end{split}
\end{align}
We call \eqref{fdouble} the ``double dose protocol.''

\subsection{Drug level statistics}\label{stats}

We now introduce some drug level statistics which we will use to compare different drug regimens and adherence rates. First, define the average drug concentration in a perfectly adherent patient (i.e.\ $p=1$),
\begin{align}\label{cp}
\C
:=\lim_{T\to\infty}\frac{1}{T}\int_{0}^{T}c^{\perf}(t)\,\dd t
=\frac{DF}{V}\frac{1}{\ke\tau},
\end{align}
where the final equality follows from \eqref{cstart}. 
We note that $\C=\AUC^{\perf}/\tau$, where $\AUC^{\perf}=DF/(V\ke)$ is the so-called ``area under the curve'' statistic commonly used in pharmacokinetics for perfect adherence \cite{gibaldi1982}. We consider $\C$ to be the desired drug level, and thus the statistics below are measured relative to $\C$.

To measure how nonadherence reduces the average drug concentration compared to $\C$, define the relative mean,
\begin{align}\label{mud}
\mu
:=\frac{1}{\C}\bigg(\lim_{T\to\infty}\frac{1}{T}\int_{0}^{T}c(t)\,\dd t\bigg).
\end{align}
To measure how the drug concentration in a patient deviates from $\C$ over time, define the {error},
\begin{align}\label{decomp}
\eps
:=\frac{1}{\C}\sqrt{\lim_{T\to\infty}\frac{1}{T}\int_{0}^{T}\big(c(t)-\C\big)^{2}\,\dd t}.
\end{align}
We note that statistics of the form \eqref{decomp} are often called relative root mean squared {error}s. 
Finally, since high drug levels may be toxic, we measure how the highest possible drug concentration compares to $\C$ via
\begin{align}\label{theta}
\theta
:=\frac{1}{\C}\Big(\sup_{t\ge0,\xi}c(t)-\C\Big).
\end{align}
Here, $\sup_{t\ge0,\xi}$ denotes the supremum over time and over all possible patterns $\xi=\{\xi_{n}\}_{n}$ of the patient remembering or forgetting. In particular, the definition of $\theta$ ensures that the drug concentration in the patient always satisfies
\begin{align*}
c(t)\le(1+\theta)\C\quad\text{for all }t\ge0.
\end{align*}

To summarize, an ideal drug regimen would generally aim to have $\mu\approx1$ and small values of $\eps$ and $\theta$. We note that aiming to have small values of $\eps$ and $\theta$ accords with the adage that ``flatter is better'' for the drug concentration time course \cite{panayiotopoulos2010}. We further note that the statistical measures $\mu$, $\eps$, and $\theta$ are all dimensionless, and they are therefore independent of the units used to measure drug amounts, volumes, time, etc. {\blue Also, $\mu$, $\eps$, and $\theta$ are symmetric functions of $\alpha:=e^{-\ke\tau}$ and $\beta:=e^{-\ka\tau}$ (meaning, $\mu$ with $\alpha=a\in(0,1)$ and $\beta=b\in(0,1)$ is equal to  $\mu$ with $\alpha=b$ and $\beta=a$, and the analogous statements for $\eps$ and $\theta$). To see this symmetry, observe that \eqref{concscale} and \eqref{cp} imply that (i) $\kappa(e^{-\ke s}-e^{-\ka s})/\C$ is a symmetric function of $\alpha$ and $\beta$ for any $s$, (ii) $c(t)$ involves a sum of terms of the form $\kappa(e^{-\ke s}-e^{-\ka s})$, and (iii) $\mu$, $\eps$, and $\theta$ only depend on $c(t)$ and $\C$ via the ratio $c(t)/\C$.}

\section{Mathematical analysis}\label{math}

We now analyze the mathematical model described above. Readers who are interested primarily in the pharmacological implications of the model may wish to skip to section~\ref{pharm}. The proofs of the theorems are in the Appendix.

\subsection{General setting}\label{general}

Let $\{\xi_{n}\}_{n\in\Z}$ be a bi-infinite, not necessarily independent, sequence of Bernoulli random variables. Fix an integer $m\ge0$, and let $\{X_{n}\}_{n\in\Z}$ be the history process,
\begin{align}\label{hist0}
X_{n}=(\xi_{n-m},\xi_{n-m+1},\dots,\xi_{n-1},\xi_{n})\in\{0,1\}^{m+1},
\end{align}
which tracks whether or not the patient remembered at dosing time $n$ and the previous $m$ dosing times. We assume that $\{X_{n}\}_{n\in\Z}$ is an irreducible, discrete-time Markov chain on the state space $\{0,1\}^{m+1}$ with time-homogeneous transition matrix \cite{norris1998}
\begin{align}\label{P}
P
=\{P(x,y)\}_{x,y\in\{0,1\}^{m+1}}
\in\R^{2^{m+1}\times2^{m+1}}.
\end{align}
The entry in the $x$-row and $y$-column of $P$ is defined by
\begin{align*}
P(x,y)
=\P(X_{1}=y\,|\,X_{0}=x),\quad x,y\in\{0,1\}^{m+1},
\end{align*}
where $x\in\{0,1\}^{m+1}$ denotes the vector,
\begin{align*}
x
=(x_{-m},x_{-m+1},\dots,x_{-1},x_{0})\in\{0,1\}^{m+1},
\end{align*}
and $y\in\{0,1\}^{m+1}$ is analogous. We assume that $\{X_{n}\}_{n\in\Z}$ is a stationary sequence and let
\begin{align*}
\pi
=\{\pi(x)\}_{x\in\{0,1\}^{m+1}}\in\R^{2^{m+1}}
\end{align*}
denote its distribution, which means
\begin{align}\label{pidef}
\pi(x)
=\P(X_{n}=x),\quad n\in\Z,\,x\in\{0,1\}^{m+1},
\end{align}
and $\pi$ satisfies the system of linear algebraic equations,
\begin{align}\label{alg1}
P^{\top}\pi=\pi,
\end{align}
where $P^{\top}$ denotes the transpose of $P$.

In the special case that $\{\xi_{n}\}_{n\in\Z}$ are independent and identically distributed (iid) with $\P(\xi_{n}=1)=p$, it is immediate that \cite{counterman2021}
\begin{align}\label{Peasy}
P(x,y)=\begin{cases}
p & \text{if }y_{0}=1,\,(x_{-m+1},\dots,x_{0})=(y_{-m},\dots,y_{-1}),\\
1-p & \text{if }y_{0}=0,\,(x_{-m+1},\dots,x_{0})=(y_{-m},\dots,y_{-1}),\\
0 & \text{otherwise},
\end{cases}
\end{align}
and
\begin{align}\label{pi}
\pi(x)
:=\P(X_{n}=x)
=p^{s(x)}(1-p)^{m+1-s(x)}>0,\quad n\in\Z,\,x\in\{0,1\}^{m+1},
\end{align}
where $s(x)\in\{0,1,\dots,m+1\}$ is the number of $1$'s in the vector $x\in\{0,1\}^{m+1}$,
\begin{align*}
s(x)
:=\sum_{k=0}^{m}x_{k}.
\end{align*}
A dosing protocol is any nonnegative function on the state space of $X_{n}$,
\begin{align}\label{f}
f:\{0,1\}^{m+1}\mapsto[0,\infty).
\end{align}

\subsection{General drug level statistics}

We now analyze the distribution of the drug concentration under the general setup of section~\ref{general}. To describe this distribution, define
\begin{align}\label{AB0}
A
:=\sum_{n=0}^{\infty}\alpha^{n}f(X_{-n}),\quad
B
:=\sum_{n=0}^{\infty}\beta^{n}f(X_{-n}),
\end{align}
where we have defined the dimensionless constants,
\begin{align}\label{alphabeta}
\alpha:=e^{-\ke{\tau}}\in(0,1),\quad
\beta:=e^{-\ka{\tau}}\in(0,1).
\end{align}
To simplify formulas, we present some results in terms of $\alpha$ and $\beta$ and some results in terms of $\ke$, $\ka$, and $\tau$.

We note that the random variables $A$ and $B$ in \eqref{AB0} generalize a class of random variables known as infinite Bernoulli convolutions. In particular, in the special case that $f$ is the single dose protocol and $\{\xi_{n}\}_{n\in\Z }$ are iid, then $A$ and $B$ are (after a linear transformation) standard infinite Bernoulli convolutions, which are known for their very irregular distributions and have been studied in the pure math literature for many decades \cite{jessen1935, kershner1935, erdos1939, peres2000, solomyak1995, peres1998, escribano2003, hu2008}.

The following theorem shows that the drug concentration $c(t)$ in \eqref{CC0} converges in distribution at large time and describes the limiting distribution in terms of $A$ and $B$ in \eqref{AB0}. The theorem also gives the large time moments of $c(t)$ in terms of the moments of $A$ and $B$, where the moments of $c(t)$ can be computed equivalently as either time averages or ensemble averages. 

\begin{theorem}[Large time drug concentration]\label{lt}
Let $c(t)$ be the drug concentration in \eqref{CC0} with $f_{n}:=f(X_{n})$. Then, for any time $t\in[0,\tau]$, we have the following convergence in distribution,
\begin{align}\label{cdt}
c(N\tau+t)
\to_{\dd}C(t)
:=\kappa\big(\alpha^{t/\tau}A-\beta^{t/\tau}B\big),\quad\text{as }N\to\infty.
\end{align}
Furthermore, for any time $t\in[0,\tau]$ and any moment $j>0$, we have that with probability one,
\begin{align}\label{mct}
\lim_{N\to\infty}\frac{1}{N}\sum_{n=0}^{N-1}\big(c(n\tau+t)\big)^{j}=
\lim_{N\to\infty}\E\big[(c(N\tau+t))^{j}\big]
=\E\big[(C(t))^{j}\big].
\end{align}
Furthermore, for any moment $j>0$, we have that with probability one,
\begin{align}
\lim_{T\to\infty}\frac{1}{T}\int_{0}^{T}\big(c(t)\big)^{j}\,\dd t
=\lim_{N\to\infty}\frac{1}{\tau}\int_{0}^{\tau}\E\big[(c(N\tau+t))^{j}\big]\,\dd t
=\frac{1}{\tau}\int_{0}^{\tau}\E\big[(C(t))^{j}\big]\,\dd t,\label{mcti1}\\
\lim_{N\to\infty}\frac{1}{N}\sum_{n=0}^{N-1}\Big(\int_{n\tau}^{(n+1)\tau}c(t)\,\dd t\Big)^{j}
=\lim_{N\to\infty}\E\Big[\Big(\int_{0}^{\tau}c(N\tau+t)\,\dd t\Big)^{j}\Big]
=\E\big[(\AUC)^{j}\big],\label{mcti2}
\end{align}
where $\AUC:=\int_{0}^{\tau}C(t)\,\dd t$.
\end{theorem}

To summarize Theorem~\ref{lt}, if a patient has been taking the drug for a long time, then the statistics of the drug concentration (either averaged over time or averaged over realizations of their stochastic adherence) can be obtained from statistics of $C(t)$. {\blue In particular, Theorem~\ref{lt} implies that the relative mean $\mu$ in \eqref{mud} can be written as the following ensemble average,
\begin{align}\label{muensemble}
\mu
=\frac{1}{\C\tau}\int_{0}^{\tau}\E[C(t)]\,\dd t.
\end{align}
Similarly, Theorem~\ref{lt} implies that the error $\eps$ in \eqref{decomp} can be written as the following ensemble average,
\begin{align}\label{epsensemble}
\eps
=\frac{1}{\C}\sqrt{\frac{1}{\tau}\int_{0}^{\tau}\Big(\E[C^{2}(t)]-2\C\E[C(t)]+\C^{2}\Big)\,\dd t}.
\end{align}}

The following theorem thus computes the first and second moments of $C(t)$. We emphasize that the theorem holds for a general dosing protocol $f$ and a general sequence $\{\xi_{n}\}_{n\in\Z}$ of Bernoulli random variables as described in section~\ref{general} (i.e.\ $\{\xi_{n}\}_{n\in\Z}$ need not be independent).

\begin{theorem}[First and second moments]\label{thm12}
The first moment of $C(t)$ is
\begin{align}\label{ob1}
\E[C(t)]
&=\kappa\big(\alpha^{t/\tau}\E[A]-\beta^{t/\tau}\E[B]\big),
\end{align}
where $\kappa$ is defined in \eqref{concscale} and
\begin{align}\label{nob1}
\E[A]
=\frac{1}{1-\alpha}\sum_{x}f(x)\pi(x),\quad
\E[B]
=\frac{1}{1-\beta}\sum_{x}f(x)\pi(x),
\end{align}
where $\sum_{x}$ denotes the sum over all $x\in\{0,1\}^{m+1}$, and $\pi$ is defined in \eqref{pidef}.

Define the vectors,
\begin{align}\label{alg2}
u_{\alpha}
:=(I-\alpha P^{\top})^{-1}v\in\R^{2^{m+1}},
\quad
u_{\beta}
:=(I-\beta P^{\top})^{-1}v\in\R^{2^{m+1}},
\end{align}
where $I\in\R^{2^{m+1}\times2^{m+1}}$ is the identity matrix, $P^{\top}$ is the transpose of $P$ in \eqref{P}, and $v\in\R^{2^{m+1}}$ is the vector with entries $v(x)=f(x)\pi(x)$ for $x\in\{0,1\}^{m+1}$. The second moment of $C(t)$ is
\begin{align}\label{ob2}
\E[(C(t))^{2}]
&=\kappa^{2}\big(\alpha^{2t/\tau}\E[A^{2}]-2\alpha^{t/\tau}\beta^{t/\tau}\E[AB]+\beta^{2t/\tau}\E[B^{2}]\big),
\end{align}
where
\begin{align}\label{nob2}
\begin{split}
\E[AB]
&=\frac{1}{1-\alpha\beta}\Big(\sum_{x}f(x)(u_{\alpha}(x)+u_{\beta}(x))-\sum_{x}(f(x))^{2}\pi(x)\Big),
\end{split}
\end{align}
and the formula for $\E[A^{2}]$ is obtained from \eqref{nob2} by replacing $\beta$ by $\alpha$, and the formula for $\E[B^{2}]$ is obtained from \eqref{nob2} by replacing $\alpha$ by $\beta$.
\end{theorem}

Theorem~\ref{thm12} can be used to obtain explicit formulas for pharmacologically relevant statistics of the drug level. To find first order statistics (i.e.\ means), one needs only to solve the linear algebraic equations in \eqref{alg1} for $\pi$ which then yields $\E[A]$ and $\E[B]${\blue, }which then yields all the statistics in Theorem~\ref{lt} for $j=1${\blue, including $\mu$ in \eqref{muensemble}}. To then find second order statistics (i.e.\ second moments, deviations, variances, etc.)\ one needs additionally only to find the matrix inverse in \eqref{alg2} to find $\E[A^{2}]$, $\E[AB]$, and $\E[B^{2}]${\blue, }which then yields all the statistics in Theorem~\ref{lt} for $j=2${\blue, including $\eps$ in \eqref{epsensemble}}.

\subsection{Simple model of nonadherence}\label{simple}

To investigate drug level statistics, we must specify the statistical patterns of patient nonadherence. Mathematically, this means that we must specify the transition matrix $P$ in \eqref{P}. For simplicity, assume that the patient remembers to take their medication at each dosing time with probability $p\in(0,1)$, independent of whether or not they remembered or forgot at prior dosing times. In particular, assume that $\{\xi_{n}\}_{n\in\Z}$ is an iid sequence with $\P(\xi_{n}=1)=p$, and thus $P$ and $\pi$ are given in \eqref{Peasy} and \eqref{pi}. In this case, Theorem~\ref{thm12} yields the following formulas. {\blue We note that we used symbolic algebra software \cite{mathematica} to obtain these formulas from Theorem~\ref{thm12}.}

\begin{corollary}\label{abf}
Assume $\{\xi_{n}\}_{n\in\Z}$ are iid with $\P(\xi_{n}=1)=p$. Then for the single dose protocol, we have that
\begin{align}
\E[A^{\single}]
&=\frac{p}{1-\alpha},
\quad
\E[B^{\single}]
=\frac{p}{1-\beta},\nonumber\\
\E[A^{\single}B^{\single}]
&=\frac{p}{1-\alpha\beta}
+p^{2}\Big(\frac{1}{(1-\alpha)(1-\beta)}-\frac{1}{1-\alpha\beta}\Big).\label{sm}
\end{align}
The formula for $\E[(A^{\single})^{2}]$ (respectively, $\E[(B^{\single})^{2}]$) is given by \eqref{sm} upon replacing $\beta$ by $\alpha$ (respectively, $\alpha$ by $\beta$). Furthermore, for the double dose protocol, we have that
\begin{align}
&\E[A^{\double}]
=\frac{p+p(1-p)}{1-\alpha},
\quad
\E[B^{\double}]
=\frac{p+p(1-p)}{1-\beta},\nonumber\\
\begin{split}
&\E[A^{\double}B^{\double}]
=\frac{4 p}{1-\alpha  \beta }
-p^2 \left(\frac{2 (\alpha +\beta )}{1-\alpha  \beta}-\frac{4}{(1-\alpha) (1-\beta)}+\frac{7}{1-\alpha  \beta}\right)\\
&\quad+\frac{p^3 \left(4 \left(2-\frac{1}{1-\alpha}-\frac{1}{1-\beta}\right)+3 \alpha +3 \beta \right)}{1-\alpha  \beta }-\frac{p^4 \left(\alpha -\frac{1}{1-\alpha}+\beta -\frac{1}{1-\beta }+2\right)}{1-\alpha  \beta }.\label{dm}
\end{split}
\end{align}
The formula for $\E[(A^{\double})^{2}]$ (respectively, $\E[(B^{\double})^{2}]$) is given by \eqref{dm} upon replacing $\beta$ by $\alpha$ (respectively, $\alpha$ by $\beta$).
\end{corollary}

Using {\blue the results of Theorem~\ref{lt} in \eqref{muensemble}-\eqref{epsensemble}, the representations in \eqref{ob1} and \eqref{ob2} in Theorem~\ref{thm12}, and Corollary~\ref{abf}}, we can calculate explicit formulas for the statistics $\mu$ and $\eps$ defined in \eqref{mud} and \eqref{decomp} in section~\ref{stats}. {\blue Again, we used symbolic algebra software \cite{mathematica} to obtain these explicit formulas from \eqref{muensemble}-\eqref{epsensemble}, \eqref{ob1} and \eqref{ob2}, and Corollary~\ref{abf}.}

\begin{corollary}\label{cm}
Assume $\{\xi_{n}\}_{n\in\Z}$ are iid with $\P(\xi_{n}=1)=p$. Using superscripts to denote the single and double dose protocols, the mean $\mu$ in \eqref{mud} is
\begin{align*}
\mu^{\single}
=p
<\mu^{\double}
=p+p(1-p),\quad\text{for all }p\in(0,1).
\end{align*}
For the single dose protocol, the error $\eps$ in \eqref{decomp} is
\begin{align*}
\eps^{\single}
=\sqrt{(\eps^{\perf})^{2}+{\eta}^{\single}},
\end{align*}
where $\eps^{\perf}$ is
\begin{align}\label{ef000}
\eps^{\perf}
&=\sqrt{\frac{\ka \ke \tau  \left(\ka \coth \left(\frac{\ke \tau }{2}\right)-\ke \coth \left(\frac{\ka \tau }{2}\right)\right)}{2 \left(\ka^2-\ke^2\right)}-1},
\end{align}
where $\coth(z):=(e^{z}+e^{-z})/(e^{z}-e^{-z})$ denotes the hyperbolic cotangent, and
\begin{align*}
{\eta}^{\single}
&=(1-p)\frac{\big(\ka^2 (4-\ke \tau  (\frac{4}{e^{\ke \tau }-1}+1))+\frac{\ka \ke^2 \tau  (e^{\ka \tau }+3)}{e^{\ka \tau }-1}-4 \ke^2\big)}{2 (\ka^2-\ke^2)}\\
&\quad+(1-p)^{2}\frac{\ka \ke \tau  (\ka e^{\ka \tau }-\ka-\ke e^{\ke \tau }+\ke)}{(\ka^2-\ke^2) (e^{\ka \tau }-1) (e^{\ke \tau }-1)}.
\end{align*}
For the double dose protocol, the error $\eps$ in \eqref{decomp} is 
\begin{align*}
\eps^{\double}
=\sqrt{(\eps^{\perf})^{2}+{\eta}^{\double}},
\end{align*}
where $\eps^{\perf}$ is in \eqref{ef000} and 
\begin{align*}
{\eta}^{\double}
=(1-p)h_{1}+(1-p)^{2}h_{2}+(1-p)^{3}h_{3}+(1-p)^{4}h_{4},
\end{align*}
and
\begin{align*}
h_{1}
&=(\ka^2-\ke^2)^{-1}\ka \ke \tau  (\ka (-e^{-\ke \tau })+\ke (e^{-\ka \tau }-1)+\ka),\\
h_{2}
&=h_{2}^{\textup{num}}[2 (\ka^2-\ke^2) \left(e^{\ka \tau }-1\right) \left(e^{\ke \tau }-1\right)]^{-1},\\
h_{3}
&=(\ka^2-\ke^2)^{-1}\ka \ke \tau  \left(\ka e^{-\ke \tau }-\ke e^{-\ka \tau }\right),\\
h_{4}
&=\frac{\ka \ke \tau  e^{-(\tau  (\ka+\ke))} \left(\ka e^{2 \ka \tau }-\ka e^{\ka \tau }+\ke e^{\ke \tau }-\ke e^{2 \ke \tau }\right)}{(\ka^2-\ke^2) \left(e^{\ka \tau }-1\right) \left(e^{\ke \tau }-1\right)},
\end{align*}
and
\begin{align*}
h_{2}^{\textup{num}}
&=e^{-(\tau  (\ka+\ke))} \big[2 \ka^2 \ke \tau  e^{\ka \tau }-2 \ka^2 \ke \tau  e^{2 \ka \tau }-2 \ka \ke^2 \tau  e^{\ke \tau }+2 \ka \ke^2 \tau  e^{2 \ke \tau }\\
&\quad+e^{\tau  (\ka+2 \ke)} \left(\ka (\ke \tau  (3 \ka-\ke)-4 \ka)+4 \ke^2\right)\\
&\quad+e^{\tau  (2 \ka+\ke)} \left(\ka (\ke \tau  (\ka-3 \ke)-4 \ka)+4 \ke^2\right)\\
&\quad-(\ka-\ke) e^{\tau  (\ka+\ke)} (\ka \ke \tau -4 (\ka+\ke))\\
&\quad-(\ka-\ke) e^{2 \tau  (\ka+\ke)} (3 \ka \ke \tau -4 (\ka+\ke))\big].
\end{align*}
\end{corollary}

In Figure~\ref{fignum}, we plot the relative mean $\mu$ and the error $\eps$ for the single and double dose protocols. The curves are the formulas in Corollary~\ref{cm} and the square markers are computed from stochastic simulations of $c(t)$ using the definitions in \eqref{mud} and \eqref{decomp} with $T=10^{6}$ and $\tau$, $D$, $F$, and $V$ set to unity. In the left plot, we set $\alpha=0.5$ and $\beta=0.25$ and $p$ varies. In the right plot, we take $\alpha=0.5$, $p=0.8$, and $\beta$ varies. This figure shows that the simulations agree with the exact analytical results obtained in Corollary~\ref{cm}.

\begin{figure}[t]
\centering
\includegraphics[width=.9\linewidth]{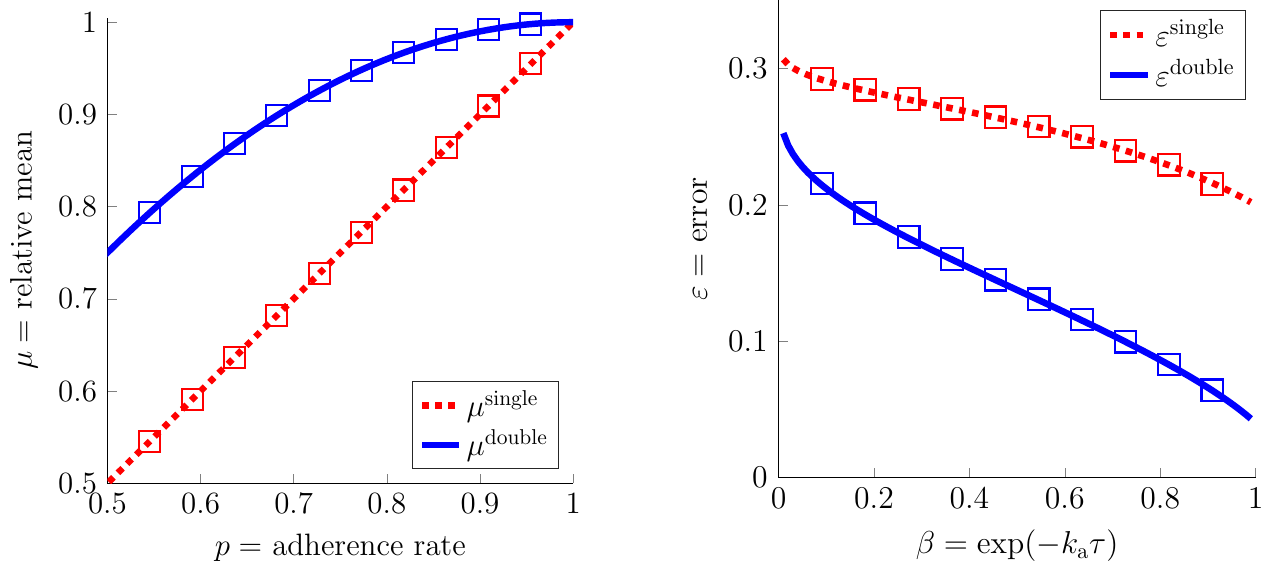}
\caption{For both the single and double dose protocols, the left panel plots the average drug concentration compared to $\C$ ($\mu$ in \eqref{mud}) and the right panel plots the error ($\eps$ in \eqref{decomp}). In both plots, the curves are the formulas in Corollary~\ref{cm} and the square markers are computed from stochastic simulations of the drug concentration. See the text for details.}
\label{fignum}
\end{figure}

\subsection{Maximum concentration}

The following theorem concerns the largest possible value of the drug concentration for the single and double dose protocols. {\blue Here, $\sup_{t\in[0,\tau],N\ge0,\xi}$ denotes the supremum over all $t\in[0,\tau]$, over all integers $N\ge0$, and over all possible patterns of patient remembering or forgetting, $\xi=\{\xi_{n}\}_{n\in\Z}\in\{0,1\}^{\Z}$.}

\begin{theorem}\label{largest}
{\blue For the single dose protocol, we have that 
\begin{align}\label{ob3}
\begin{split}
\sup_{t\in[0,\tau],N\ge0}c^{\perf}(N\tau+t)
&=\sup_{t\in[0,\tau],N\ge0,\xi}c^{\single}(N\tau+t)\\
&=\sup_{t\in[0,\tau],\xi}C^{\single}(t)
=\sup_{t\in[0,\tau]}C^{\perf}(t)
=C^{\perf}(t^{*}),
\end{split}
\end{align}
where $C^{\perf}(t)=\kappa(\alpha^{t/\tau}/(1-\alpha)-\beta^{t/\tau}/(1-\beta))$ and
\begin{align}\label{ob4}
t^{*}
:={\frac{\tau\ln \Big(\frac{(1-\alpha) \ln (\beta )}{(1-\beta ) \ln (\alpha )}\Big)}{\ln (\alpha /\beta )}}.
\end{align}
Therefore, the maximum relative overshoot in \eqref{theta} for the single dose protocol is
\begin{align}\label{tst}
\theta^{\single}
=\theta^{\perf}
=\frac{\ka \ke \tau}{\ka-\ke}  \Bigg(\frac{\left(\frac{\ke (e^{\ka \tau }-1)}{\ka (e^{\ke \tau }-1)}\right)^{\frac{\ke}{\ka-\ke}}}{e^{\ke \tau }-1}-\frac{\left(\frac{\ke (e^{\ka \tau }-1)}{\ka (e^{\ke \tau }-1)}\right)^{\frac{\ka}{\ka-\ke}}}{e^{\ka \tau }-1}\Bigg)-1>0.
\end{align}

For the double dose protocol, we have the upper bound
\begin{align}\label{nob3new}
\begin{split}
\sup_{t\in[0,\tau],N\ge0,\xi}c^{\double}(N\tau+t)
=
\sup_{t\in[0,\tau],\xi}C^{\double}(t)
&\le C^{\perf}(t^{*})+b(s^{*}),
\end{split}
\end{align}
where $C^{\perf}(t^{*})$ is given above and $b(s^{*})$ is the maximum concentration obtained after a single dose. Specifically, $b(s)=\kappa(\alpha^{s/\tau}-\beta^{s/\tau})$ and
\begin{align}\label{sstar}
s^{*}
=\frac{\tau\ln(\ln(\beta)/\ln(\alpha))}{\ln(\alpha/\beta)}
=\frac{\ln(\ka/\ke)}{\ka-\ke}>0.
\end{align}
Therefore, using that $b(s^{*})/\C=\tau\ka^{\frac{\ke}{\ke-\ka}}\ke^{\frac{\ka}{\ka-\ke}}$, the maximum relative overshoot in \eqref{theta} for the double dose protocol, $\theta^{\double}$, has the upper bound
\begin{align}\label{tdt}
\theta^{\double}
\le
\theta^{\perf}+\tau\ka^{\frac{\ke}{\ke-\ka}}\ke^{\frac{\ka}{\ka-\ke}}.
\end{align}}
\end{theorem}

\section{Some pharmacological implications}\label{pharm}

In this section, we explore some pharmacological implications of the mathematical model and analysis in sections~\ref{model}-\ref{math}. We compare drug regimens using the statistics $\mu$, $\eps$, and $\theta$ defined in \eqref{mud}-\eqref{theta}. Recall that $\ka$ and $\ke$ denote the absorption and elimination rates, and $\tau$ denotes the dosing interval. We present some results in terms of the dimensionless parameters,
\begin{align*}
\alpha
:=e^{-\ke\tau}\in(0,1),\quad
\beta
:=e^{-\ka\tau}\in(0,1).
\end{align*}
The adherence rate is $p\in(0,1]$, and for simplicity, we consider the model of nonadherence of section~\ref{simple} in which the patient remembers or forgets to take their medication independently at different dosing times. However, we note that the actual doses taken by the patient are not independent if they follow the double dose protocol, since in this case they take a double dose if they remember and they happened to have forgotten at the prior dosing time.

\subsection{Perfect adherence}

To better understand the effects of nonadherence, we first consider the case of perfect adherence (i.e.\ $p=1$). We have that $\mu^{\perf}=1$ and Corollary~\ref{cm} yields that the error $\eps$ in \eqref{decomp} is 
\begin{align}\label{ef}
\eps^{\perf}
&=\sqrt{\frac{\ka \ke \tau  \left(\ka \coth \left(\frac{\ke \tau }{2}\right)-\ke \coth \left(\frac{\ka \tau }{2}\right)\right)}{2 \left(\ka^2-\ke^2\right)}-1},
\end{align}
where $\coth(z):=(e^{z}+e^{-z})/(e^{z}-e^{-z})$ denotes the hyperbolic cotangent. By differentiating \eqref{ef}, it is straightforward to check that $\eps^{\perf}$ is an increasing function of $\ke$, $\ka$, and $\tau$,
\begin{align*}
\frac{\partial}{\partial\ke}\eps^{\perf}>0,\quad
\frac{\partial}{\partial\ka}\eps^{\perf}>0,\quad
\frac{\partial}{\partial\tau}\eps^{\perf}>0.
\end{align*}
Furthermore, \eqref{ef} implies that $\eps^{\perf}$ has the limiting values,
\begin{align*}
\lim_{\alpha\to0,\beta\to0}\eps^{\perf}
=\infty,\quad
\lim_{\alpha\to1}\eps^{\perf}
=\lim_{\beta\to1}\eps^{\perf}
&=0.
\end{align*}
Finally, {\blue Theorem~\ref{largest} yields that} the maximum fluctuation above $\C$ for a perfectly adherent patient is
\begin{align}\label{tp}
\theta^{\perf}
=\frac{\ka \ke \tau}{\ka-\ke}  \Bigg(\frac{\left(\frac{\ke (e^{\ka \tau }-1)}{\ka (e^{\ke \tau }-1)}\right)^{\frac{\ke}{\ka-\ke}}}{e^{\ke \tau }-1}-\frac{\left(\frac{\ke (e^{\ka \tau }-1)}{\ka (e^{\ke \tau }-1)}\right)^{\frac{\ka}{\ka-\ke}}}{e^{\ka \tau }-1}\Bigg)-1>0,
\end{align}
which, like the {error} $\eps^{\perf}$, is an increasing function of $\ke$, $\ka$, and $\tau$.

\begin{figure}[t]
\centering
\includegraphics[width=1\linewidth]{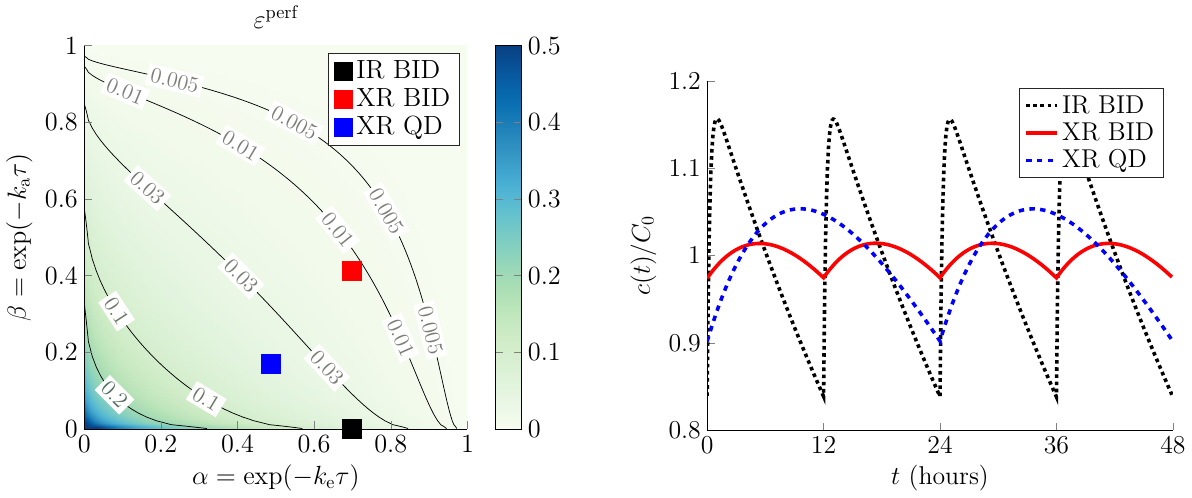}
\caption{The left panel is a contour plot of the {error} for perfect adherence, $\eps^{\perf}$ in \eqref{ef}, as a function of $\alpha$ and $\beta$. The three square markers correspond to the antiepileptic drug lamotrigine in its IR version dosed twice-daily (BID) and its XR version dosed either BID or once-daily (QD). The right panel plots drug concentration time courses corresponding to these three drug regimens. See the text for details.}
\label{figperf}
\end{figure}

To setup some analysis below, we consider these results in the context of IR and XR drugs. Recall that IR and XR formulations of a drug have identical elimination rates, but the XR formulation has a much slower absorption rate and is sometimes prescribed with a larger dosing interval. In the left panel of Figure~\ref{figperf}, we show a contour plot of $\eps^{\perf}$ as a function of $\alpha$ and $\beta$. In this plot, the three square markers correspond to a twice-daily (BID) IR antiepileptic drug and its XR formulation dosed either BID or once-daily (QD). The specific drug is lamotrigine with elimination rate $\ke=0.03\,\textup{hr}^{-1}$ and the IR and XR absorption rates are $\ka^{\textup{IR}}=3\,\textup{hr}^{-1}\gg\ka^{\textup{XR}}=0.07\,\textup{hr}^{-1}$ {\blue (parameter values were obtained in \cite{chen2013})}. The values of $\eps^{\perf}$ for the three square markers are
\begin{align}\label{ai}
\eps^{\perf}_{\textup{IR,BID}}
\approx0.1,\quad
\eps^{\perf}_{\textup{XR,BID}}
\approx0.01,\quad
\eps^{\perf}_{\textup{XR,QD}}
\approx0.05.
\end{align}
Hence, \eqref{ai} implies that the {error}s from $\C$ are largest for IR with BID dosing and smallest for XR with BID dosing. This is illustrated in the right panel of Figure~\ref{figperf}, which plots concentration time courses for these three dosing regimens.

\subsection{Average drug concentration}

We now analyze the effects of nonadherence, beginning with the average drug concentration. Corollary~\ref{cm} yields that $\mu$ in \eqref{mud} for the single and double dose protocols is
\begin{align*}
\mu^{\single}=p
<\mu^{\double}=p+p(1-p),\quad p\in(0,1).
\end{align*}
Therefore, the average drug concentrations in a patient following the single or double dose protocols are
\begin{align}
\begin{split}\label{mm4}
\lim_{T\to\infty}\frac{1}{T}\int_{0}^{T}c^{\single}(t)\,\dd t
&
=p\C
=p\frac{DF}{V}\frac{1}{\ke\tau},\\
\lim_{T\to\infty}\frac{1}{T}\int_{0}^{T}c^{\double}(t)\,\dd t
&
=\big[p+p(1-p)\big]\C
=\big[p+p(1-p)\big]\frac{DF}{V}\frac{1}{\ke\tau}.
\end{split}
\end{align}
There are two implications of \eqref{mm4} which we emphasize. First, \eqref{mm4} ensures that the long time average drug concentration is independent of the absorption rate $\ka$. Hence, for any given adherence $p\in(0,1]$, IR and XR formulations of the same drug result in the same average drug concentration in the patient. This fact will be useful in section~\ref{slower} below.

Second, \eqref{mm4} quantifies how the double dose protocol increases the average drug level compared to the single dose protocol. In the left panel of Figure~\ref{fignum}, we plot $\mu^{\single}$ and $\mu^{\double}$ as functions of the adherence $p\in(0,1)$, which shows that the increase in drug levels obtained from switching from the single to double dose protocol is substantial for common values of $p$. Indeed, if one is only concerned with the average drug level, then following the double dose protocol with adherence of only $p=0.8$ is equivalent to following the single dose protocol with adherence $p=0.8+0.8\times0.2=0.96$. In practical terms for a QD drug, a patient following the double dose protocol who tends to miss one dose a week ($p=6/7$) has roughly the same average drug levels as a patient following the single dose protocol who tends to miss only one dose a month ($6/7+6/7\times1/7\approx29/30$).

\subsection{Slow absorption or elimination reduces the error}\label{slower}

\begin{figure}[t]
\centering
\includegraphics[width=1\linewidth]{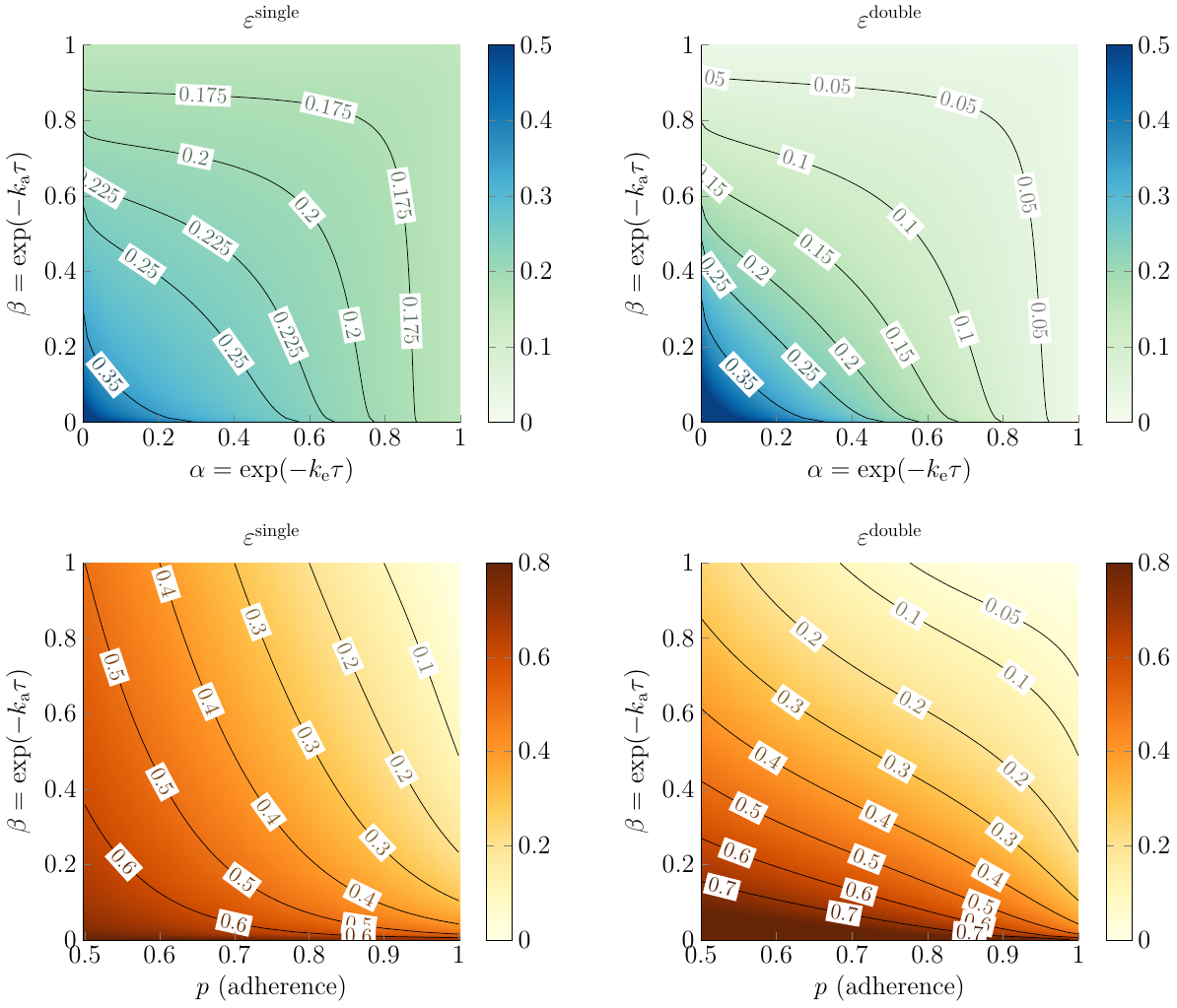}
\caption{{\blue Contour plots of the {error} $\eps$ in \eqref{decomp}, using the formulas in Corollary~\ref{cm}. In the top two panels, $\alpha=e^{-\ke\tau}$ and $\beta=e^{-\ka\tau}$ vary and $p=0.85$. In the bottom two panels, $p$ and $\beta$ vary and $\alpha=0.01$. Since $\eps$ is symmetric in $\alpha$ and $\beta$, the bottom two panels can also be interpreted as fixing $\beta=0.01$ and letting $\alpha$ vary along the vertical axis.
}}
\label{figdev}
\end{figure}

Corollary~\ref{cm} gives explicit formulas for the error $\eps$ in \eqref{decomp} for both the single and double dose protocols. 
In {\blue the top two panels of} Figure~\ref{figdev}, we use these formulas to produce contour plots of the error $\eps$ as a function of $\alpha$ and $\beta$ for the single and double protocols with $p=0.85$. These plots show that $\eps$ is a decreasing function of $\alpha$ and $\beta$. Indeed, using the formulas in Corollary~\ref{cm}, we have verified with extensive numerical tests that for both the single and double dose protocols, $\frac{\partial}{\partial\alpha}\eps<0$ and $\frac{\partial}{\partial\beta}\eps<0$ for any choice of $\alpha,\beta,p\in(0,1)$ with $\alpha\neq\beta$. Since $\alpha=\exp(-\ke\tau)$ and $\beta=\exp(-\ka\tau)$, this implies that for both the single and double dose protocols, 
\begin{align}\label{signs}
\frac{\partial}{\partial\ke}\eps>0,\quad
\frac{\partial}{\partial\ka}\eps>0,\quad
\frac{\partial}{\partial\tau}\eps>0.
\end{align}
Hence, we can decrease the error for any adherence rate $p$ by decreasing the absorption rate $\ka$, the elimination rate $\ke$, or the dosing interval $\tau$.

{\blue
In the bottom two panels of Figure~\ref{figdev}, we plot the errors $\eps^{\single}$ and $\eps^{\double}$ as functions of $p$ and $\beta$ with $\alpha=0.01$ fixed (since $\eps$ is a symmetric function of $\alpha$ and $\beta$, this is equivalent to fixing $\beta=0.01$ and varying $\alpha$). These plots show that increasing $\beta$ (or $\alpha$) effectively increases the adherence, in the sense that a small $p$ and a large $\beta$ can yield the same error as a large $p$ and a small $\beta$. For example, the contour line at $\eps^{\single}=0.3$ in the bottom left panel shows that $p\approx0.75$ and $\beta\approx0.7$ yields the same error as $p\approx1$ and $\beta\approx0.1$. 
}
 
These results have implications for prescribing IR versus XR drug formulations, since \eqref{signs} implies that an XR drug yields smaller errors than an IR drug, as long as the two formulations have the same dosing interval. Furthermore, \eqref{mm4} ensures that the average drug concentration in the patient is independent of the absorption rate. Hence, replacing an IR drug with an XR drug with the same dosing interval (i) reduces fluctuations in the drug concentration and (ii) does not change the average drug concentration. In addition, \eqref{signs} implies that an XR drug with a longer dosing interval can increase or decrease the {error} compared to an IR drug with a shorter dosing interval, since ${{{\eps}}}$ is an increasing function of both $\ka$ and $\tau$.

\begin{figure}[t]
\centering
\includegraphics[width=1\linewidth]{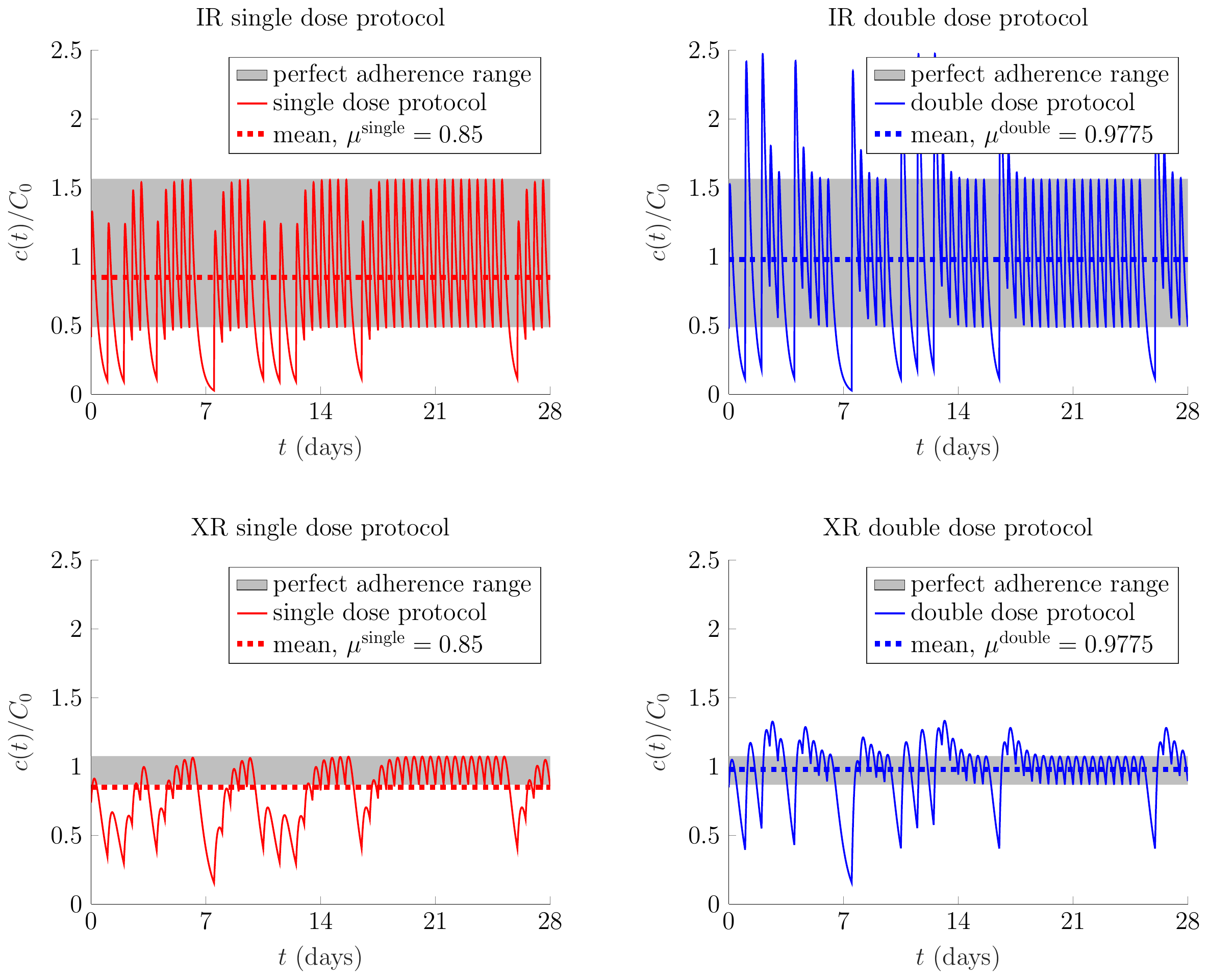}
\caption{Stochastic simulations of the drug concentration time course for 4 weeks of twice-daily dosing of Quetiapine fumarate under four regimens. The gray regions show the range between the peaks of troughs of the drug concentration under perfect adherence. The dashed flat lines describe the average drug concentration. The top two panels are for the IR formulation, and the bottom two panels are for the XR formulation. The left two panels are for the single dose protocol and the right two panels are for the double dose protocol.}
\label{figsim}
\end{figure}

To illustrate what the values of $\eps^{\single}$ and $\eps^{\double}$ imply about IR and XR formulations and the single and double dose protocol for an imperfectly adherent patient, in Figure~\ref{figsim} we plot stochastic simulations of the drug concentration time course under various scenarios of taking Quetiapine fumarate, which is an antipsychotic drug. The top two panels of Figure~\ref{figsim} are for the IR version and the bottom two panels are for the XR version. Further, the left two panels are for the single dose protocol and the right two panels are for the double dose protocol. We set $p=0.85$ and take $\tau=12\,\text{hr}$, $\ke=0.12\,\text{hr}^{-1}$, and $\ka^{\IR}=1.46\,\text{hr}^{-1}$ and $\ka^{\XR}=0.1\,\text{hr}^{-1}$ for the respective IR and XR formulations of Quetiapine fumarate {\blue (parameter values were obtained in \cite{elkomy2020})}. We note that the particular random sequence of missed doses is identical in all four panels. 

Figure~\ref{figsim} shows that the fluctuations are much greater for the IR formulation (top two panels) compared to the XR formulation (bottom two panels), which reflects the fact that the corresponding values of the error $\eps$ are
\begin{align}\label{dvals}
\eps^{\double}_{\textup{XR}}\approx0.21
<\eps^{\single}_{\textup{XR}}\approx0.26
<\eps^{\single}_{\textup{IR}}\approx0.44
<\eps^{\double}_{\textup{IR}}\approx0.50.
\end{align}
Furthermore, this plot shows that (a) for the IR formulation, the single dose protocol has smaller fluctuations than the double dose protocol, whereas (b) for the XR formulation, the double dose protocol has smaller fluctuations than the single dose protocol. Both of these points accord with the theoretical values in \eqref{dvals}, which were computed from the formulas in Corollary~\ref{cm}. In fact, for this drug, the {error} for the XR formulation with the double dose protocol and imperfect adherence ($p=0.85$), is less than the IR formulation with perfect adherence,
\begin{align*}
\eps^{\double}_{\textup{XR}}\approx0.21
<\eps^{\perf}_{\textup{IR}}\approx0.34,
\end{align*}
which is evident from Figure~\ref{figsim} by comparing the blue curve in the bottom right panel to the gray region in the top panels.

\subsection{Double dose protocol mitigates nonadherence for drugs with slow absorption or elimination}\label{dvs}

When is the double dose protocol preferable to the single dose protocol? To compare the single and double dose {error}s, in Figure~\ref{figratio} we show contour plots of the relative difference,
\begin{align}\label{rd}
\frac{\eps^{\single}-\eps^{\double}}{\eps^{\single}},
\end{align}
as a function of $\alpha$ and $\beta$ for $p=0.9$ (top left panel) and $p=0.7$ (top right panel). The plots show that the double dose protocol has a smaller error than the single dose protocol except in the case that both $\alpha$ and $\beta$ are small. This is illustrated in Figure~\ref{figsim}, where the single dose protocol has smaller fluctuations for the IR formulation and the double dose protocol has smaller fluctuations for the XR formulation.

One possible concern regarding the double dose protocol is whether it might cause the drug concentration to rise far above the desired average $\C$. To address this question, we investigate $\theta$ in \eqref{theta}, which is the maximum possible relative increase above $\C$. For the single dose protocol, this maximum is achieved by a perfectly adherent patient, and thus $\theta^{\single}=\theta^{\perf}$ where $\theta^{\perf}$ is in \eqref{tp}. For the double dose protocol, {\blue we} are not able to find an explicit formula for $\theta^{\double}$, but {\blue Theorem~\ref{largest} gives the following upper bound,
\begin{align}\label{ublater}
\theta^{\double}
\le\theta^{\perf}+\tau\ka^{\frac{\ke}{\ke-\ka}}\ke^{\frac{\ka}{\ka-\ke}}.
\end{align}
I}n Figure~\ref{figratio}, we show contour plots of $\theta^{\single}$ in the bottom left panel and {\blue the upper bound for $\theta^{\double}$ in \eqref{ublater} in the bottom right panel.} These plots show that both $\theta^{\single}$ and $\theta^{\double}$ vanish for large values of $\alpha$ or $\beta$. Furthermore, $\theta^{\double}$ is relatively small as long as $\alpha$ and $\beta$ are not both small. Hence, $\theta^{\double}$ is small in the same parameter regime in which the double dose protocol has a smaller error $\eps$ than the single dose protocol.

\begin{figure}[t]
\centering
\includegraphics[width=1\linewidth]{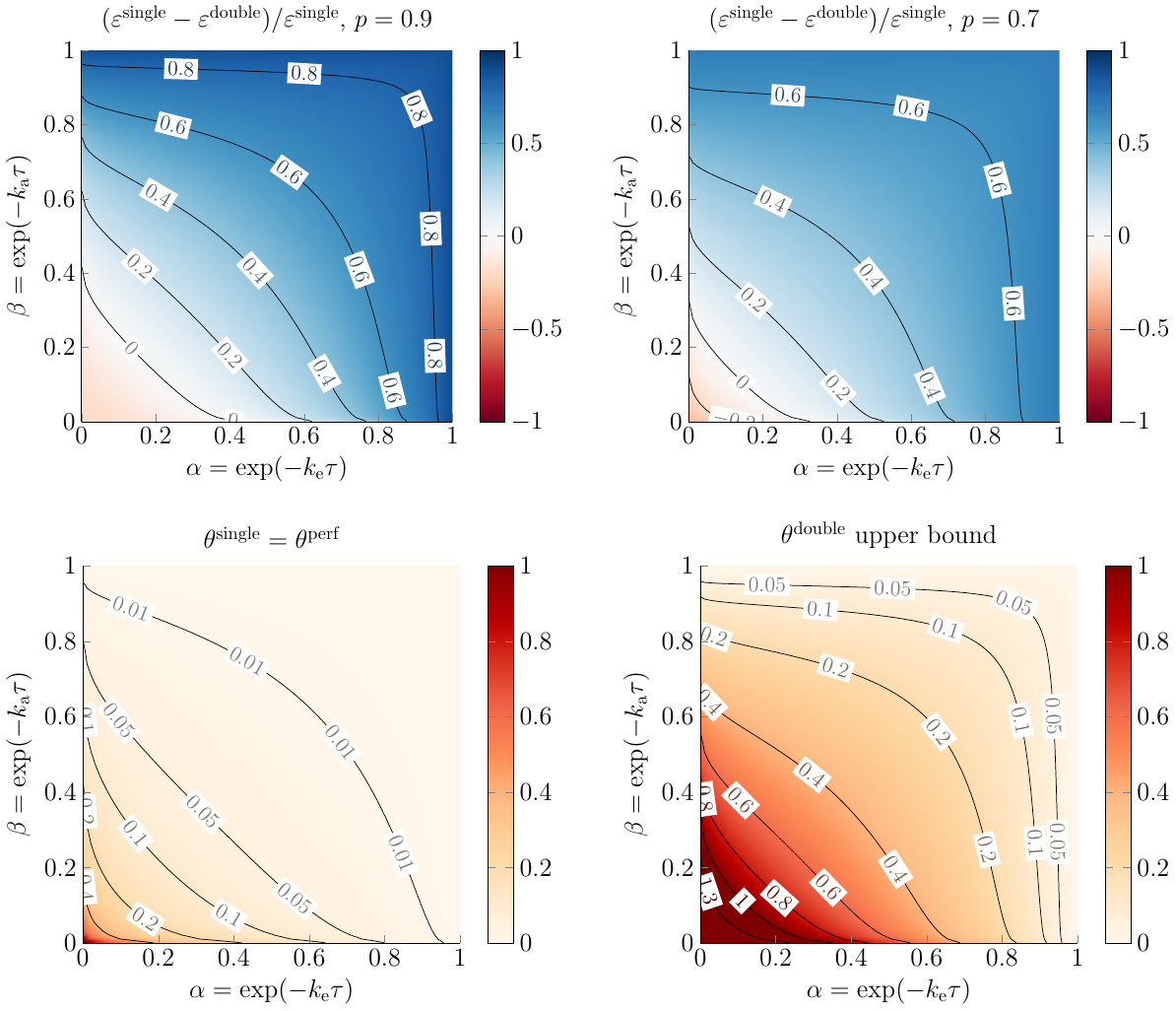}
\caption{The top two panels are contour plots of the relative difference in {error}s for the single and double dose {error}s in \eqref{rd} for $p=0.9$ (top left panel) and $p=0.7$ (top right panel). The bottom {\blue left panel is a contour plot of $\theta^{\single}$ (see \eqref{theta} and \eqref{tst} in Theorem~\ref{largest}) for the single dose protocol. The bottom right panel is a contour plot of the upper bound for $\theta^{\double}$ given in \eqref{tdt} in Theorem~\ref{largest} (and also in \eqref{ublater}).}}
\label{figratio}
\end{figure}

\subsection{Comparing drug regimens and designing XR drugs}

As we have shown above, the {error} $\eps$ is an increasing function of both $\ka$ and $\tau$. Therefore, the {error} $\eps^{\XR}$ of an XR regimen with absorption rate $\ka^{\XR}$ and dosing interval $\tau^{\XR}$ may or may not be larger than the {error} $\eps^{\IR}$ of an IR regimen with absorption rate $\ka^{\IR}$ and $\tau^{\IR}$, depending on the particular parameter values. XR drugs are often made to simplify dosing, so suppose that the IR drug is dosed twice-daily and the XR drug is dosed once-daily, and thus
\begin{align}\label{dis}
\tau^{\IR}=12\,\text{hr},\quad \tau^{\XR}=24\,\text{hr}.
\end{align}
How slow does the XR absorption rate $\ka^{\XR}$ need to be so that the XR regimen has a smaller {error} than the IR regimen? This is a natural question both for the design of XR drugs and for comparing specific XR and IR regimens.

In the left panel of Figure~\ref{figkp}, we plot the value of $\ka^{\XR}$ needed so that ${{{\eps}}}^{\XR}={{{\eps}}}^{\IR}$ as a function of $\ka^{\IR}$. In this plot, the dosing intervals are as in \eqref{dis} and we take $\ke=0.03\,\text{hr}^{-1}$, which corresponds to the antiepileptic drug lamotrigine described above. The three curves are for perfect adherence (i.e.\ $p=1$) and the single and double dose protocols for adherence $p=0.85$. Notice that the curves describing imperfect adherence lie below the curve for perfect adherence. This means that if one takes into account imperfect adherence, then a slower XR absorption rate is needed so that the IR and XR regimens have the same {error}. The black circle marks where lamotrigine lies in the $(\ka^{\IR},\ka^{\XR})$ plane \cite{chen2013}.

The discussion above assumes that the IR regimen adherence $p^{\IR}$ is the same as the XR regimen adherence $p^{\XR}$ (i.e.\ $p^{\IR}=p^{\XR}=p$). However, the simplified dosing allowed by XR versions is often aimed at improving adherence. While our model cannot predict how adherence might increase by switching from an IR to an XR regimen, our model can predict the value of $p^{\XR}$ needed so that ${{{\eps}}}^{\XR}={{{\eps}}}^{\IR}$ for a given value of $p^{\IR}$. In the right panel of Figure~\ref{figkp}, we plot this value of $p^{\XR}$ as a function of $p^{\IR}$, where the dosing intervals are as in \eqref{dis} and the drugs are lamotrigine ($\ke=0.03\,\textup{hr}^{-1}$, $\ka^{\textup{IR}}=3\,\textup{hr}^{-1}$, $\ka^{\textup{XR}}=0.07\,\textup{hr}^{-1}$ \cite{chen2013}), a hypothetical antiepileptic drug (AED) considered in \cite{pellock2016} ($\ke=0.08\,\text{hr}^{-1}$, $\ka^{\IR}=5\,\text{hr}^{-1}$, $\ka^{\XR}=0.5\,\text{hr}^{-1}$), and Quetiapine fumarate ($\ke=0.12\,\text{hr}^{-1}$, $\ka^{\IR}=1.46\,\text{hr}^{-1}$, $\ka^{\XR}=0.1\,\text{hr}^{-1}$ \cite{elkomy2020}). This plot shows that, for lamotrigine, $p^{\XR}$ is only slightly larger than $p^{\IR}$ for most values of $p^{\IR}$, and $p^{\XR}$ is actually less than $p^{\IR}$ for large values of $p^{\IR}$. Further, $p^{\XR}$ is much higher than $p^{\IR}$ for the hypothetical antiepileptic drug and $p^{\XR}$ is much lower than $p^{\IR}$ for Quetiapine fumarate. Hence, depending on parameter values, switching from a BID IR regimen to a QD XR regimen may or may not require an increase in adherence to have the same {error}.

\begin{figure}[t]
\centering
\includegraphics[width=1\linewidth]{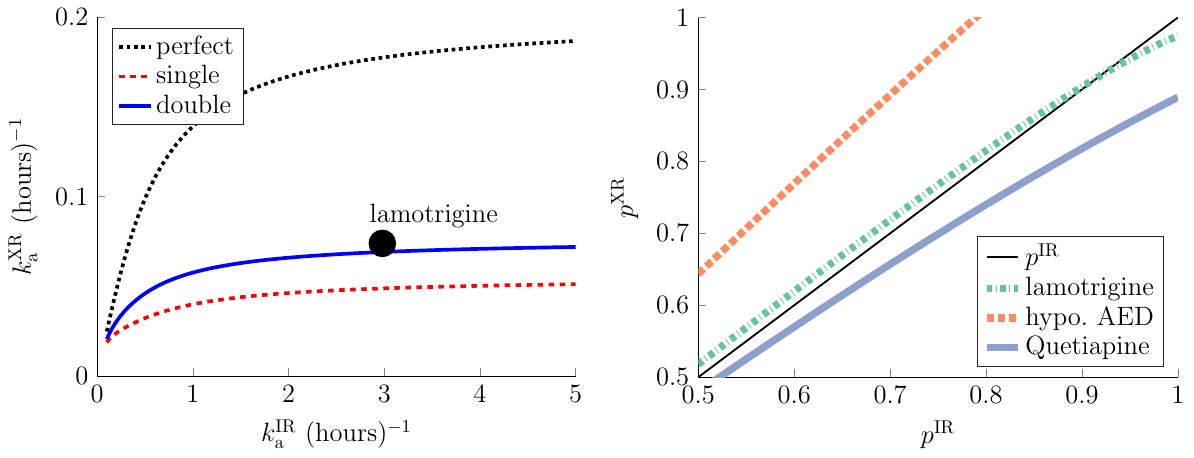}
\caption{The left panel plots the XR absorption rate needed so that a once-daily XR regimen has the same {error} as a twice-daily IR regimen, as a function of the IR absorption rate. For three different drugs, the right panel plots the once-daily XR adherence rate needed so that a once-daily XR regimen has the same {error} as a twice-daily IR regimen, as a function of the twice-daily IR adherence rate. See the text for details. }
\label{figkp}
\end{figure}

\section{Discussion}

There are at least three major hurdles which hinder the study of medication nonadherence. First, clinical trials which force patients to miss doses of {\blue the medication being tested could be unethical \cite{millum2013}}. Second, nonadherence is by nature erratic, as patients do not miss doses in precise patterns. Third, there are many parameters (adherence rates, absorption and elimination rates, dosing intervals, etc.)\  to vary in any systematic investigation, and it is difficult to disentangle the individual contributions of each of these parameters. For all of these reasons, probabilistic modeling and analysis is an important tool for studying and mitigating nonadherence.

In this paper, we formulated and analyzed a mathematical model of the drugs levels in an imperfectly adherent patient. The model assumes that the patient misses doses randomly, and we used stochastic analysis to determine pharmacologically relevant statistics of the drug levels in the patient. This analysis revealed several principles for designing drug regimens to mitigate nonadherence and provided tools to predict the efficacy of different regimens when challenged by nonadherence. Of particular note, we showed the resiliency of XR drugs to nonadherence compared to IR drugs when they are dosed at the same frequency. We further showed the benefit of taking a double dose following a missed dose if the absorption or elimination rate is slow compared to the dosing interval.

XR drugs are often recommended in order to improve adherence because they permit less frequent dosing (i.e.\ increase the adherence $p$ by increasing the dosing interval $\tau$). Recently, it has been argued that using XR drugs for this particular purpose does not justify the additional monetary cost of XR versus IR formulations \cite{sumarsono2020}. Indeed, empirical studies of the increase in $p$ associated with increasing $\tau$ have found mixed results \cite{udelson2009, spencer2011, ackloo2008}. However, we have demonstrated an alternative benefit of XR formulations described above (namely, that they better mitigate nonadherence when dosed at the same frequency as IR formulations). It would be interesting to carry out an economic cost/benefit analysis of prescribing XR formulations at the same frequency of IR formulations. 

{\blue While there are ethical issues regarding clinical trials which force patients to miss doses of medication \cite{millum2013}, some of our results could be tested in clinical trials and inform clinical trial design. For example, the efficacy of IR versus XR drug formulations dosed at the same frequency could be compared in a clinical trial to test our predictions, since some trial participants invariably miss or delay doses of their own accord (i.e.\ without compulsion). Indeed, while adherence is generally thought to be higher in clinical trials compared to clinical practice \cite{lowy2011}, poor adherence is nevertheless a significant problem in clinical trials that corrupts estimates of the benefits and risks of a medicine \cite{breckenridge2017}. In addition, our predicted benefit of taking a double dose following a missed dose for drugs with slow absorption or elimination kinetics could also be tested in a clinical trial. For example, one group of participants could be instructed to skip any missed dose and another group could be instructed to take a double dose following a missed dose. Such a study could be especially useful if it included electronic monitoring of adherence \cite{burnier2019}.}

Several prior works have used probabilistic models to study nonadherence. In seminal work, Li and Nekka \cite{li2007} formulated stochastic pharmacokinetic models in which the successive times between doses are independent and identically distributed random variables. Other authors have developed stochastic pharmacokinetic models that allow for variability in dosing times, dose sizes, and eliminations rates for the case of intravenous administration \cite{levy2013} and oral administration \cite{fermin2017}. The discrete time model in \cite{fermin2017} is very similar to our model in the case of the single dose protocol. These prior investigations did not consider different ways of handling missed doses. Assuming that the patient never misses two or more doses consecutively, Ma \cite{ma2017} studied how different ways of handling missed doses affect the time it takes the drug concentration in the patient to enter a specified therapeutic window. The model in the present work generalizes our previous model in \cite{counterman2021} which took the drug absorption rate to be infinite. 

Many groups have used numerical simulations of computational models to study medication nonadherence \cite{boissel2002, garnett2003, reed2004, dutta2006, ding2012, chen2013, gidal2014, brittain2015, stauffer2017, sunkaraneni2018, hard2018, elkomy2020, gu2020}. Our results are in general agreement with some of the conclusions of this prior computational work which investigated some specific drugs, including (i) the efficacy of a double dose following a missed dose and (ii) that adherence thresholds should depend on specifics of the drug regimen.

To elaborate on (ii), the adherence rate $p=0.8$ has long been deemed the threshold separating ``adherent'' and ``nonadherent'' patients \cite{burnier2019, haynes1976}. However, the actual drug level statistics in the patient depend on many parameters other than $p$. We now illustrate this concretely in our model for the IR and XR formulations of Quetiapine fumarate. Suppose patient \#1 and patient \#2 have respective adherence rates
\begin{align*}
p_{1}
=0.9
>p_{2}
=0.7,
\end{align*}
so that patient \#1 would be considered ``adherent'' and patient \#2 would be considered ``nonadherent.'' However, suppose patient \#1 is on the IR formulation and skips missed doses (i.e.\ the single dose protocol) and patient \#2 is on the XR formulation and takes a double dose following a missed dose (i.e.\ the double dose protocol), both with prescribed twice-daily dosing (i.e.\ $\tau=12\,\text{hr}$). Then, if $\C$ is the average drug concentration in a perfectly adherent patient (see \eqref{cp}), then the average drug concentration in patient \#1 is
\begin{align*}
p_{1}\C=0.9\C
\end{align*}
which is actually less than the average drug concentration in patient \#2, which is
\begin{align*}
(p_{2}+p_{2}(1-p_{2}))\C
=0.91\C.
\end{align*}
Furthermore, using the values of $\ke$, $\ka^{\IR}$, and $\ka^{\XR}$ for Quetiapine fumarate {\blue obtained in \cite{elkomy2020} (see section~\ref{pharm} above)}, the respective {error}s for patient \#1 and patient \#2 are
\begin{align*}
{{{\eps}}}_{1}
\approx0.40
>{{{\eps}}}_{2}
\approx0.31.
\end{align*}
Hence, by following the double dose protocol with an XR formulation rather than the single dose protocol with an IR formulation, a supposedly ``nonadherent'' patient can actually have more efficacious drug levels than a supposedly ``adherent'' patient.

\section{Appendix}

In this appendix, we prove the theorems of section~\ref{math}.

\begin{proof}[Proof of Theorem~\ref{lt}]
For any dosing protocol $f$ and any integers $N\ge M$, define
\begin{align}\label{recall1}
A_{M,N}
&:=
\sum_{n=M}^{N}\alpha^{N-n}f(X_{n}),\quad
B_{M,N}
:=
\sum_{n=M}^{N}\beta^{N-n}f(X_{n}).
\end{align}
For $t\in[0,\tau]$, define the random variable
\begin{align}\label{CMN}
C_{M,N}(t)
:=\kappa(\alpha^{t/\tau}A_{M,N}-\beta^{t/\tau}B_{M,N}).
\end{align}
Further, for any $N\in\Z$, define the almost sure limits,
\begin{align}\label{recall2}
\begin{split}
A_{-\infty,N}
&:=\lim_{M\to-\infty}A_{M,N},\quad
B_{-\infty,N}
:=\lim_{M\to-\infty}B_{M,N},\\
C_{-\infty,N}(t)
&:=\lim_{M\to-\infty}C_{M,N}(t),
\end{split}
\end{align}
and notice that $A_{-\infty,0}=A$, $B_{-\infty,0}=B$, and $C_{-\infty,0}(t)=C(t)$ where $A$ and $B$ are defined in \eqref{AB0} and $C(t)$ is defined in \eqref{cdt}. {\blue To see why the almost sure convergence in \eqref{recall2} is guaranteed, note first that the function $f$ must be bounded since its domain $\{0,1\}^{m+1}$ is finite. Hence, the terms in the sums in \eqref{recall1} are bounded by the terms in a geometric series, and so the Weierstrass M-test ensures the almost sure convergence in \eqref{recall2}.}

Notice that the drug concentration in \eqref{CC0} with $f_{n}=f(X_{n})$ can be written as
\begin{align*}
c({t})
&=\kappa\Big(\alpha^{{t}/\tau-N({t})}\sum_{n=0}^{N({t})}\alpha^{N(t)-n}f_{n}-\beta^{{t}/\tau-N({t})}\sum_{n=0}^{N({t})}\beta^{N(t)-n}f_{n}\Big),\quad t\ge0.
\end{align*}
Therefore, 
\begin{align}\label{same}
c(N\tau+t)
=C_{0,N}(t),\quad\text{for any $t\in[0,\tau]$ and integer $N\ge0$}.
\end{align}
Since $\{X_{n}\}_{n\in\Z}$ is stationary, it follows that
\begin{align}\label{shift}
C_{0,N}(t)
=_{\dd}C_{-N,0}(t),\quad\text{for any integer }N\ge 0,
\end{align}
where $=_{\dd}$ denotes equality in distribution. {\blue Now, as in \eqref{recall2}, we have that}
\begin{align}\label{Y0}
\begin{split}
\lim_{N\to\infty}C_{-N,0}(t)
=C(t)
=\kappa(\alpha^{t/\tau}A-\beta^{t/\tau}B),\quad\text{for }t\in[0,\tau].
\end{split}
\end{align}
Equations \eqref{same}, \eqref{shift}, and \eqref{Y0} yield \eqref{cdt}. We note that random variables akin to \eqref{Y0} are sometimes called random pullback attractors because they take an initial condition and pull it back to the infinite past \cite{Crauel01, Mattingly99, Schmalfuss96, lawley15sima, lawley2019hhg}. 

Since $f$ is bounded, $C_{0,N}(t)$ can be bounded by a deterministic constant independent of $N$, and thus \eqref{shift}, \eqref{Y0}, and the bounded convergence theorem yield 
\begin{align}\label{momentconvergence}
\E[(C_{0,N}(t))^{j}]
\to\E[({C(t)})^{j}],\quad\text{as }N\to\infty\text{ for all }j>0.
\end{align}
Combining \eqref{momentconvergence} with \eqref{same} and \eqref{shift} yields the second equality in \eqref{mct}. Combining the second equality in \eqref{mct} with the bounded convergence theorem yields the second equality in \eqref{mcti1}.

Finally, the same argument that gave the second equality in \eqref{mct} gives the second equality in \eqref{mcti2} upon noticing that \eqref{same}, \eqref{shift}, and \eqref{Y0} all still hold when integrated from $t=0$ to $t=\tau$.

To obtain the first equalities in \eqref{mct}-\eqref{mcti2} for the time averages, we first note that Theorem~7.1.3 in \cite{durrett2019} ensures that $\{C_{-\infty,n}(t)\}_{n\in\Z}$ is ergodic for any fixed $t\ge0$ since $\{X_{n}\}_{n\in\Z}$ is ergodic and stationary. We thus have that
\begin{align*}
\lim_{N\to\infty}\frac{1}{N}\sum_{n=0}^{N-1}(C_{-\infty,n}(t))^{j}
=\E[(C_{-\infty,0}(t))^{j}]
=\E[(C(t))^{j}]\quad\text{with probability one},
\end{align*}
by Birkhoff's ergodic theorem (see, for example, Theorem~7.2.1 in \cite{durrett2019}). By \eqref{same}, we have that for any $t\in[0,\tau]$, 
\begin{align*}
\frac{1}{N}\sum_{n=0}^{N-1}(c(n\tau+t))^{j}
=\frac{1}{N}\sum_{n=0}^{N-1}(C_{0,n}(t))^{j}.
\end{align*}
Hence, in order to prove the first equality in \eqref{mct}, it remains to prove
\begin{align}\label{wts21}
\lim_{N\to\infty}\frac{1}{N}\sum_{n=0}^{N-1}(C_{-\infty,n}(t))^{j}
=\lim_{N\to\infty}\frac{1}{N}\sum_{n=0}^{N-1}(C_{0,n}(t))^{j}\quad\text{with probability one}.
\end{align}
To prove this, we first note the bound,
\begin{align*}
C_{-\infty,n}(t)-C_{0,n}(t)
=\kappa\sum_{i=-\infty}^{-1}(\alpha^{t/\tau+n-i}-\beta^{t/\tau+n-i})f(X_{i})
&\le\frac{2\kappa f_{+}(\max\{\alpha,\beta\})^{n+1}}{1-\max\{\alpha,\beta\}}\\
&=:K\gamma^{n},
\end{align*}
where $f_{+}:=\sup_{x\in\{0,1\}^{m+1}}f(x)<\infty$ since $\{0,1\}^{m+1}$ is finite.

If $j\ge1$ is an integer, then the binomial theorem implies the general identity for $a,b\in\R$,
\begin{align*}
(a+b)^{j}=a^{j}+b\sum_{k=1}^{j}{j\choose k}b^{k-1}a^{j-k}.
\end{align*}
Therefore, 
\begin{align}\label{sca}
(C_{-\infty,n}(t))^{j}
=(C_{0,n}(t)+C_{-\infty,n}(t)-C_{0,n}(t))^{j}
\le(C_{0,n}(t))^{j}+K_{0}\gamma^{n},
\end{align}
for a suitably chosen deterministic constant $K_{0}$ independent of $n$. Since $f$ is nonnegative, it follows that
\begin{align}\label{sub}
(C_{-\infty,n}(t))^{j}
-K_{0}\gamma^{n}
\le(C_{0,n}(t))^{j}
\le(C_{-\infty,n}(t))^{j},
\end{align}
which then yields \eqref{wts21} for the case that $j\ge1$ is an integer.

If $j>0$ is not an integer, then notice that the function $g(a)=a^{j-\floor*{j}}$ for $a>0$ is concave and therefore subadditive, and thus
\begin{align}\label{sca2}
\begin{split}
(C_{-\infty,n}(t))^{j-\floor*{j}}
&\le(C_{0,n}(t))^{j-\floor*{j}}
+(C_{-\infty,n}(t)-C_{0,n}(t))^{j-\floor*{j}}\\
&\le(C_{0,n}(t))^{j-\floor*{j}}
+(K\gamma^{n})^{j-\floor*{j}}.
\end{split}
\end{align}
Therefore, \eqref{sca} and \eqref{sca2} imply
\begin{align*}
(C_{-\infty,n}(t))^{j}
&\le\big[(C_{0,n}(t))^{j-\floor*{j}}
+(K\gamma^{n})^{j-\floor*{j}}\big]\big[(C_{0,n}(t))^{{\floor*{j}}}+K_{0}\gamma^{n}\big]\\
&\le(C_{0,n}(t))^{j}+K_{1}\gamma_{1}^{n},
\end{align*}
for suitably chosen deterministic constants $K_{1}>0$ and $\gamma_{1}\in(0,1)$ which are independent of $n$. Hence, \eqref{sub} holds with $K_{0}$ replaced by $K_{1}$ and $\gamma$ replaced by $\gamma_{1}$ and \eqref{wts21} follows.

Having proven the first equality in \eqref{mct}, the first equality in \eqref{mcti1} then follows from the bounded convergence theorem upon noting that
\begin{align*}
\lim_{T\to\infty}\frac{1}{T}\int_{0}^{T}(c(t))^{j}\,\dd t
&=\lim_{N\to\infty}\frac{1}{N\tau}\int_{0}^{N\tau}(c(t))^{j}\,\dd t\\
&=\lim_{N\to\infty}\frac{1}{N}\sum_{n=0}^{N-1}\frac{1}{\tau}\int_{n\tau}^{(n+1)\tau}(c(t))^{j}\,\dd t.
\end{align*}
Finally, the first equality in \eqref{mcti2} follows from applying the same argument that gave the first equality in \eqref{mct} to the integrals $\int_{0}^{\tau}C_{-\infty,n}(t)\,\dd t$ and $\int_{0}^{\tau}C_{0,n}(t)\,\dd t$.
\end{proof}


\begin{proof}[Proof of Theorem~\ref{thm12}]
Equations~\eqref{ob1} and \eqref{ob2} follow immediately from the definition of $C(t)$ in \eqref{Y0}. It thus remains to prove \eqref{nob1} and \eqref{nob2}, {\blue which generalizes} the proof of Theorem~1 in \cite{counterman2021}. Recalling the definitions in \eqref{recall1} and \eqref{recall2}, notice that
\begin{align}\label{inv}
{\A}
=_{\dd}A_{-\infty,1}
=\alpha {\A}+f(X_{1}),\quad
{B}
=_{\dd}B_{-\infty,1}
=\beta {B}+f(X_{1}),
\end{align}
where $=_{\dd}$ denotes equality in distribution.  Taking the expectation of \eqref{inv} and rearranging yields \eqref{nob1}.

To obtain $\E[A^{2}]$, we first square \eqref{inv}, take expectation, and rearrange to obtain
\begin{align}\label{tree0}
\E[A^{2}]
=\frac{1}{1-\alpha^{2}}\Big(2\alpha\E\big[{\A}f(X_{1})\big]+\E\big[(f(X_{1}))^{2}\big]\Big).
\end{align}
By definition of expectation, we have that $\E\big[(f(X_{1}))^{2}\big]
=\sum_{x}(f(x))^{2}\pi(x)$. 
To compute $\E[{\A}f(X_{1})]$, let $1_{E}\in\{0,1\}$ denote the indicator function on an event $E$, which means $1_{E}=1$ if $E$ occurs and $1_{E}=0$ otherwise. Thus,
\begin{align}\label{tree1}
\E[{\A}f(X_{1})]
=\sum_{x}f(x)\E[{\A}1_{X_{1}=x}].
\end{align}
Multiplying \eqref{inv} by $1_{X_{1}=x}$, taking expectation, and using that $({\A},X_{0})$ is equal in distribution to $(A_{-\infty,1},X_{1})$ yields
\begin{align}\label{c0}
\E[{\A}1_{X_{0}=x}]
=\E[A_{-\infty,1}1_{X_{1}=x}]
=\alpha\E[{\A}1_{X_{1}=x}]
+f(x)\pi(x),\quad x\in\{0,1\}^{m+1}.
\end{align}
The conditional expectation tower property (Theorem 5.1.6 in \cite{durrett2019}) implies
\begin{align}\label{c1}
\E[{\A}1_{X_{1}=x}]
=\sum_{y}\E[{\A}1_{X_{1}=x}1_{X_{0}=y}]
=\sum_{y}\E[{\A}1_{X_{0}=y}]P(y,x),
\end{align}
where $P$ is in \eqref{P}. Combining \eqref{c0} and \eqref{c1} yields the following system of linear algebraic equations for $\E[{\A}1_{X_{0}=x}]$,
\begin{align}\label{c2}
\E[{\A}1_{X_{0}=x}]
=\alpha\sum_{y}\E[{\A}1_{X_{0}=y}]P(y,x)+f(x)\pi(x),\quad x\in\{0,1\}^{m+1}.
\end{align}
If we define the vectors ${{u}}_{\alpha}\in\R^{2^{m+1}}$ and $v\in\R^{2^{m+1}}$ by
\begin{align*}
{{u}}_{\alpha}(x)
&:=\E[{\A}1_{X_{0}=x}],\quad
v(x)
:=f(x)\pi(x),
\end{align*}
then \eqref{alg2} solves \eqref{c2}. We note that the Perron-Frobenius theorem guarantees the invertibility of $I-\alpha P^{\top}$ since $I-\alpha P^{\top}=\alpha(\alpha^{-1}I-P^{\top})$ and $\alpha\in(0,1)$. Therefore, \eqref{c0} implies
\begin{align}\label{jg}
\E[Af(X_{1})]
=\sum_{x}f(x)\E[A1_{X_{1}=x}]
=\sum_{x}f(x)\frac{1}{\alpha}(u_{\alpha}(x)-f(x)\pi(x)).
\end{align}
Combining \eqref{jg} with \eqref{tree0} yields the formula for $\E[A^{2}]$ given by \eqref{nob2} upon replacing $\beta$ by $\alpha$. The analogous argument yields $\E[B^{2}]$ (given by \eqref{nob2} upon replacing $\alpha$ by $\beta$). To obtain $\E[AB]$, we first observe that
\begin{align}\label{jg2}
AB
=_{\dd}A_{-\infty,1}B_{-\infty,1}
=(\alpha A+f(X_{1}))(\beta B+f(X_{1})).
\end{align}
Taking the expectation of \eqref{jg2} and rearranging yields
\begin{align*}
\E[AB]
=\frac{1}{1-\alpha\beta}\Big(\alpha\E\big[{\A}f(X_{1})\big]+\beta\E\big[{B}f(X_{1})\big]+\E\big[(f(X_{1}))^{2}\big]\Big).
\end{align*}
Using \eqref{jg} and the analogous equation for $\E\big[{B}f(X_{1})\big]$ yields the formula for $\E[AB]$ in \eqref{nob2} and completes the proof.
\end{proof}


\begin{proof}[Proof of Theorem~\ref{largest}]
{\blue
Since missing doses can only decrease the concentration for the single dose protocol, we have the following pair of inequalities,
\begin{align}
\sup_{t\in[0,\tau],N\ge0,\xi}c^{\single}(N\tau+t)
&\le\sup_{t\in[0,\tau],N\ge0}c^{\perf}(N\tau+t),\label{ineq00}\\
\sup_{t\in[0,\tau],\xi}C^{\single}(t)
&\le\sup_{t\in[0,\tau]}C^{\perf}(t).\label{ineq01}
\end{align}
But, setting $\xi_{n}=1$ for all $n$ yields $c^{\single}(N\tau+t)=c^{\perf}(N\tau+t)$ and $C^{\single}(t)=C^{\perf}(t)$, and thus the inequalities in \eqref{ineq00} and \eqref{ineq01} can be replaced by equalities. Hence, we have obtained the first and third equalities in \eqref{ob3}. 

The second equality in \eqref{ob3} and the first equality in \eqref{nob3new} follow from the convergence in distribution in \eqref{cdt} in Theorem~\ref{lt}.
To see this, note first that for any dosing protocol, we have by \eqref{same} and \eqref{shift} that
\begin{align}\label{sameshift}
c(N\tau+t)
=C_{0,N}(t)
=_{\dd}C_{-N,0}(t)
\le C(t),\quad\text{for any $t\in[0,\tau]$, $N\ge0$},
\end{align}
where $C_{M,N}(t)$ is defined in \eqref{CMN} and $C(t):=C_{-\infty,0}(t)$ is defined in \eqref{recall2}. The inequality in \eqref{sameshift} holds because $f$ is nonnegative. The $=_{\dd}$ in \eqref{sameshift} denotes equality in distribution, where the probability measure $\P$ on the set of sequences $\xi=\{\xi_{n}\}_{n\in\Z}$ can be any measure as described in section~\ref{general}. In particular, we can take $\P$ to be the probability measure for the case that $\{\xi_{n}\}_{n\in\Z}$ are iid with $\P(\xi_{n}=1)=p\in(0,1)$ as in section~\ref{simple}. The important point is that for this choice of $\P$, a supremum over sequences $\xi$ is the same as an essential supremum over sequences $\xi$ (since for any finite sequence $\{\zeta_{i}\}_{i=1}^{M}\in\{0,1\}^{M}$ and any $n\in\Z$, we have $\P(\xi_{n}=\zeta_{1},\dots,\xi_{n+M}=\zeta_{M})>0$). Therefore, \eqref{sameshift} implies that
\begin{align}\label{ineq02}
\sup_{t\in[0,\tau],N\ge0,\xi}c(N\tau+t)
\le\sup_{t\in[0,\tau],\xi}C(t).
\end{align}
To obtain equality in \eqref{ineq02}, fix $t\in[0,\tau]$ and let $\delta>0$. By definition of supremum,
\begin{align*}
\P\big(C(t)>\sup_{\xi}C(t)-\delta\big)>0,
\end{align*}
where we again take $\P$ to be as in section~\ref{simple} so that a supremum over $\xi$ and an essential supremum over $\xi$ are equivalent. Hence, the convergence in distribution in \eqref{cdt} in Theorem~\ref{lt} implies that we can take $N$ sufficiently large so that
\begin{align*}
\P\big(c(N\tau+t)>\sup_{\xi}C(t)-\delta\big)>0.
\end{align*}
Since $\delta>0$ is arbitrary, we thus have that
\begin{align}\label{ineq99}
\sup_{N\ge0,\xi}c(N\tau+t)
\ge\sup_{\xi}C(t).
\end{align}
Since $t\in[0,\tau]$ is arbitrary, combining \eqref{ineq99} and \eqref{ineq02} implies that the inequality in \eqref{ineq02} can be replaced by equality. Since this holds for any dosing protocol, we have obtained the second equality in \eqref{ob3} and the first equality in \eqref{nob3new}.

The final equality in \eqref{ob3} and the maximizing time $t^{*}$ in \eqref{ob4} follow from a simple calculus exercise. The formula in \eqref{tst} follows from combining \eqref{ob3}-\eqref{ob4} with the definition of $\theta$ in \eqref{theta}.

We now prove the inequality in \eqref{nob3new}. Fix $t\in[0,\tau]$ and recall that $C^{\double}(t)$ is
\begin{align}\label{cdp}
C^{\double}(t)
=\kappa(\alpha^{t/\tau}A^{\double}-\beta^{t/\tau}B^{\double})
=\sum_{n=0}^{\infty}K_{n}(t)g_{n},
\end{align}
where $K_{n}(t)$ is the coefficient,
\begin{align*}
K_{n}(t)
:=\kappa(\alpha^{n+t/\tau}-\beta^{n+t/\tau})\ge0,\quad n\ge0,
\end{align*}
and
\begin{align}\label{gn}
g_{n}:=f^{\double}(X_{-n}),\quad n\ge0.
\end{align}
By definition of the double dose protocol, we have that for all $n\ge0$, 
\begin{align}\label{dddp}
\textup{$g_{n}\in\{0,1,2\}$ and
if $g_{n}=2$, then $g_{n+1}=0$}.
\end{align}

We claim that there is a finite, nonnegative integer $n^{*}\ge0$ such that 
\begin{align}\label{claim0}
\begin{split}
K_{n}(t)
&< K_{n+1}(t)\quad\text{if }0\le n\le n^{*}-1,\\
K_{n}(t)
&\ge K_{n+1}(t)\quad\text{if }n\ge n^{*}.
\end{split}
\end{align}
That is, the sequence $\{K_{n}(t)\}_{n\ge0}$ is strictly increasing in $n$ for $n<n^{*}$ and nonincreasing in $n$ for $n>n^{*}$. To prove this claim, we momentarily treat $n$ as a continuous variable and differentiate $K_{n}(t)$ with respect to $n$,
\begin{align}\label{dk}
\frac{\partial}{\partial n}K_{n}(t)
=\kappa(\alpha^{n+t/\tau}\ln\alpha
-\beta^{n+t/\tau}\ln\beta).
\end{align}
Rearranging \eqref{dk} shows that $\frac{\partial}{\partial n}K_{n}(t)=0$ if and only if
\begin{align*}
n=
n_{0}
:=
\frac{\ln(\ln\beta/\ln\alpha)}{\ln(\alpha/\beta)}-t/\tau\in\R.
\end{align*}
Note further that the second derivative is negative at $n_{0}$,
\begin{align*}
\frac{\partial^{2}}{\partial n^{2}}K_{n}(t)\Big|_{n=n_{0}}
&=\kappa(\alpha^{n+t/\tau}(\ln\alpha)^{2}
-\beta^{n+t/\tau}(\ln\beta)^{2})\Big|_{n=n_{0}}\\
&=\kappa\beta^{n_{0}+t/\tau}(\ln\beta)(\ln\alpha-\ln\beta)<0.
\end{align*}
Hence, $K_{n}(t)\le K_{n_{0}}(t)$ for all $n\in\R$. Therefore, if $n_{0}\le0$, then the claim is satisfied with $n^{*}=0$. If $n_{0}>0$, then the claim is satisfied by either $n^{*}=\floor{n_{0}}\ge0$ or $n^{*}=\ceil{n_{0}}\ge1$, where $\floor{\cdot}$ and $\ceil{\cdot}$ denote the floor and ceiling functions, respectively. To distinguish between the case $n^{*}=\floor{n_{0}}\ge0$ or $n^{*}=\ceil{n_{0}}\ge1$ for $n_{0}>0$, one merely checks if $K_{\floor{n_{0}}}(t)<K_{\ceil{n_{0}}}(t)$ or $K_{\floor{n_{0}}}(t)>K_{\ceil{n_{0}}}(t)$. We note that if $K_{\floor{n_{0}}}(t)=K_{\ceil{n_{0}}}(t)$, then we can simply take $n^{*}=\floor{n_{0}}$. Therefore, we have verified the claim.

We now claim that if $C^{\double}(t)$ is to be maximized, then we must have that
\begin{align}\label{claim2}
g_{n}=1\quad\text{for all }0\le n\le n^{*}-1.
\end{align}
The claim in \eqref{claim2} is vacuously true if $n^{*}=0$, so suppose $n^{*}\ge1$. To prove the claim in \eqref{claim2}, we start with $n=0$. It is immediate that the maximizing value must either be $g_{0}=1$ or $g_{0}=2$, since setting $g_{0}=0$ only makes the first term in \eqref{cdp} smaller compared to if $g_{0}=1$ or $g_{0}=2$, and it does not allow any other term to be larger than if $g_{0}=1$. If $g_{0}=2$, then we must set $g_{1}=0$ by \eqref{dddp}. However, we claim that $C^{\double}(t)$ is certainly larger if $g_{0}=g_{1}=1$ compared to if $g_{0}=2$ and $g_{1}=0$. To see this, note first that the values of $g_{i}$ for $i\ge2$ are unconstrained by either choice. Further, since $n^{*}\ge1$, \eqref{claim0} implies that $K_{0}(t)<K_{1}(t)$ and thus
\begin{align*}
K_{0}(t)\cdot2+K_{1}(t)\cdot0
< K_{0}(t)\cdot1+K_{1}(t)\cdot1.
\end{align*}
Therefore, $C^{\double}(t)$ is larger if $g_{0}=g_{1}=1$ rather than $g_{0}=2$ and $g_{1}=0$. At this point in the argument, it is still not determined if $C^{\double}(t)$ is larger by taking $g_{0}=g_{1}=1$ or $g_{0}=1$ and $g_{1}=2$. Nevertheless, we conclude that $C^{\double}(t)$ is maximized by setting $g_{0}=1$ in this case that $n^{*}\ge1$. Repeating this argument shows that we must take $g_{n}=1$ for all $n<n^{*}$ in order to maximize $C^{\double}(t)$, and thus we have verified the claim in \eqref{claim2}.

We further claim that if $C^{\double}(t)$ is to be maximized, then 
\begin{align}\label{claim9}
\textup{if $g_{n}=0$, then $g_{n-1}=2$},\quad\text{for all }n\ge1.
\end{align}
To see why \eqref{claim9} holds, observe that if $g_{n}=0$ and $g_{n-1}\neq 2$ for some $n\ge1$, then one could change the value of $g_{n}$ to be $g_{n}=1$ without changing the value of $g_{i}$ for any $i\neq n$, and this would make the value of $C^{\double}(t)$ larger. Hence, \eqref{claim9} holds for any sequence $\{g_{n}\}_{n\ge0}$ which maximizes $C^{\double}(t)$.

Now, let $\{g_{n}\}_{n\ge0}$ be any sequence as in \eqref{gn}-\eqref{dddp} that satisfies \eqref{claim2} and \eqref{claim9}. We claim that the corresponding value of $C^{\double}(t)$ in \eqref{cdp} satisfies
\begin{align}\label{claim3}
C^{\double}(t)
\le K_{n'}(t)+\sum_{n=0}^{\infty}K_{n}(t)
=K_{n'}(t)+C^{\perf}(t),
\end{align}
where $n'\ge n^{*}$ is the smallest integer such that $g_{n}=2$,
\begin{align}\label{nprime}
n'
:=\inf\{n\ge0:g_{n}=2\}\ge n^{*},
\end{align}
where we set $K_{n'}(t)=0$ if $n'=\infty$ in the case that $g_{n}=1$ for all $n\ge0$ (note that \eqref{claim3} is trivially satisfied in this case). In words, the claim in \eqref{claim3} means that the concentration for the double dose protocol is always less than the concentration for perfect adherence plus the concentration from a single dose taken $n'$ dosing times in the past. Note that \eqref{claim9} and \eqref{nprime} imply that
\begin{align}\label{onesbefore}
g_{n}=1,\quad\text{for all }0\le n\le n'-1.
\end{align}
Using the definition of $C^{\double}(t)$ in \eqref{cdp}, \eqref{onesbefore}, and the definition of $n'$ in \eqref{nprime}, the claim in \eqref{claim3} is equivalent to
\begin{align}\label{claim3e}
K_{n'}(t)+\sum_{n=0}^{\infty}K_{n}(t)(1-g_{n})
=\sum_{n=n'+1}^{\infty}K_{n}(t)(1-g_{n})
\ge0.
\end{align}
Define the sequence of $\{S_{n}\}_{n\ge n'+1}$ by
\begin{align*}
S_{n}
:=\sum_{i=n'+1}^{n-1}(1-g_{i}),\quad n\ge n'+1,
\end{align*}
where $S_{n'+1}=0$. Note that since $g_{n}\in\{0,1,2\}$, we are assured that
\begin{align}\label{s11}
S_{n}-S_{n-1}
=1-g_{n-1}
\in\{-1,0,1\}.
\end{align}
Further, \eqref{dddp} implies that 
\begin{align}\label{s22}
\textup{if $S_{n}-S_{n-1}=-1$, then $S_{n+1}-S_{n}=1$}.
\end{align}
In addition, since \eqref{nprime} implies that $g_{n'}=2$, \eqref{dddp} implies
\begin{align}\label{s33}
S_{n'+2}=1-g_{n'+1}=1.
\end{align}
In words, \eqref{s11} means that successive terms in the sequence $\{S_{n}\}_{n\ge n'+1}$ can change by at most $\pm1$, and \eqref{s22} means that if successive terms decrease by 1, then the next term increases by 1. Since $S_{n'+1}=0$ and $S_{n'+2}=1$ by \eqref{s33}, we conclude that
\begin{align}\label{spos}
S_{n}\ge0\quad \text{for all }n\ge n'+1.
\end{align}

Using summation by parts, we have that
\begin{align}\label{sbp}
\begin{split}
\sum_{n=n'+1}^{\infty}K_{n}(t)(1-g_{n})
&=\sum_{n=n'+1}^{\infty}K_{n}(t)(S_{n+1}-S_{n})\\
&=\sum_{n=n'+2}^{\infty}S_{n}(K_{n-1}(t)-K_{n}(t))\ge0,
\end{split}
\end{align}
by \eqref{spos} and the fact that $n'\ge n^{*}$ and $K_{n}(t)$ is nonincreasing in $n$ for $n\ge n^{*}$ as in \eqref{claim0}. In \eqref{sbp}, we also used that $\lim_{n\to\infty}K_{n}(t)=0$ and $S_{n'+1}=0$. Hence, we have verified \eqref{claim3e} and thus \eqref{claim3}.

Therefore, \eqref{claim3} implies that
\begin{align*}
\sup_{t\in[0,\tau],\xi}C^{\double}(t)
&\le \sup_{t\in[0,\tau],n'\ge0}\Big(K_{n'}(t)+\sum_{n=0}^{\infty}K_{n}(t)\Big)\\
&= \sup_{t\in[0,\tau],n'\ge0}\Big(K_{n'}(t)+C^{\perf}(t)\Big)\\
&\le \sup_{t\in[0,\tau],n'\ge0}K_{n'}(t)+\sup_{t\in[0,\tau]}C^{\perf}(t).
\end{align*}
We obtained $\sup_{t\in[0,\tau]}C^{\perf}(t)$ in \eqref{ob3}-\eqref{ob4}, and a straightforward calculus exercise yields
\begin{align*}
\sup_{t\in[0,\tau],n'\ge0}K_{n'}(t)
=\sup_{s\ge0}\kappa(\alpha^{s/\tau}-\beta^{s/\tau})
=\kappa(\alpha^{s^{*}/\tau}-\beta^{s^{*}/\tau}),
\end{align*}
where $s^{*}$ is given in \eqref{sstar}. The bound in \eqref{tdt} follows from combining \eqref{nob3new} and \eqref{sstar} with the definition of $\theta$ in \eqref{theta}.
}
\end{proof}


\subsubsection*{Acknowledgments}
SDL was supported by the National Science Foundation (Grant Nos.\ DMS-1944574 and DMS-1814832). 



\bibliography{library.bib}
\bibliographystyle{unsrt}

\end{document}